\documentclass[conference]{IEEEtran}

\usepackage{amsmath,amssymb,amsfonts,amsthm}
\usepackage{mathtools}
\usepackage{algorithm}
\usepackage[noend]{algpseudocode}
\usepackage{enumitem}

\usepackage{xcolor}
\definecolor{linkblue}{RGB}{1,1,190}
\definecolor{citegreen}{RGB}{1,190,1}
\usepackage{hyperref}
\hypersetup{linkcolor=linkblue
           ,citecolor=citegreen}
\usepackage[ocgcolorlinks]{ocgx2}
\usepackage{cleveref}

\newtheorem{definition}{Definition}
\crefname{definition}{Definition}{Definitions}
\newtheorem{lemma}[definition]{Lemma}
\crefname{lemma}{Lemma}{Lemmas}
\newtheorem{theorem}[definition]{Theorem}
\crefname{theorem}{Theorem}{Theorems}
\newtheorem{proposition}[definition]{Proposition}
\crefname{proposition}{Proposition}{Propositions}
\newtheorem{corollary}[definition]{Corollary}
\crefname{corollary}{Corollary}{Corollaries}

\theoremstyle{definition}
\newtheorem{remark}[definition]{Remark}
\crefname{remark}{Remark}{Remarks}
\newtheorem{example}[definition]{Example}
\crefname{example}{Example}{Examples}

\crefname{algorithm}{Algorithm}{Algorithms}

\setlist[enumerate,1]{label=\arabic*)}

\usepackage{microtype}

\DeclareMathOperator{\GL}{GL}
\DeclareMathOperator{\lspan}{span}
\DeclareMathOperator{\rank}{rank}
\DeclareMathOperator{\im}{im}
\DeclareMathOperator{\grank}{\overline{r}}

\DeclareMathOperator{\chr}{char}
\DeclareMathOperator{\val}{\mathsf v}
\DeclarePairedDelimiter{\card}{\lvert}{\rvert}

\newcommand{\bcdot}{\boldsymbol{\cdot}}

\newcommand{\bC}{\mathbb C}
\newcommand{\bZ}{\mathbb Z}
\newcommand{\bQ}{\mathbb Q}

\newcommand{\cA}{\mathcal A}

\newcommand{\cS}{\mathcal S}
\newcommand{\cT}{\mathcal T}

\newcommand{\QQbar}{\overline{\mathbb Q}}

\newcommand{\quo}[1]{\mathbf{q}({#1})}

\newcommand{\zcl}[1]{\widetilde{#1}}

\newcommand{\comp}[1]{\textsf{c}({#1})}

\newcommand{\algc}[1]{{#1}^{\textup{alg}}}
\newcommand{\alg}[1]{{#1}^{\textup{alg}}}

\newcommand{\defit}[1]{\emph{{#1}}}

\begin{document}

\title{Computing the linear hull: Deciding Deterministic? and Unambiguous? for weighted automata over fields}

\author{\IEEEauthorblockN{Jason P. Bell}
  \IEEEauthorblockA{\textit{Department of Pure Mathematics}\\
    \textit{University of Waterloo}\\
    Waterloo, Canada\\
    jpbell@uwaterloo.ca}
  \thanks{J.\,P.~Bell was supported by NSERC grant RGPIN-2022-02951.}
\and
\IEEEauthorblockN{Daniel Smertnig}
\IEEEauthorblockA{\textit{Institute for Mathematics and Scientific Computing}\\
  \textit{NAWI Graz, University of Graz}\\
  Graz, Austria\\
  daniel.smertnig@uni-graz.at}}

\maketitle

\begin{abstract}
  The \emph{\textup{(}left\textup{)} linear hull} of a weighted automaton over a field is a topological invariant.
  If the automaton is minimal, the linear hull can be used to determine whether or not the automaton is equivalent to a deterministic one.
  Furthermore, the linear hull can also be used to determine whether the minimal automaton is equivalent to an unambiguous one.
  We show how to compute the linear hull, and thus prove that it is decidable whether or not a given automaton over a number field is equivalent to a deterministic one.
  In this case we are also able to compute an equivalent deterministic automaton.
  We also show the analogous decidability and computability result for the unambiguous case.
  Our results resolve a problem posed in a 2006 survey by Lombardy and Sakarovitch.
\end{abstract}


\section{Introduction}

Every unweighted (finite) automaton is equivalent to a \emph{deterministic} automaton\footnote{deterministic automata also called \emph{sequential} or \emph{subsequential} \cite[Remark V.1.2]{sakarovitch09} automata in the weighted case}, and there is a determinization procedure to find such an automaton.
For automata with weights in a semiring $K$ (in short, $K$-automata), this is no longer true.
More generally, a $K$-automaton is \defit{unambiguous} if (i) between each two states $p$ and $q$ and for every word $w$ there is at most one path from $p$ to $q$ labeled by $w$, and (ii) every word has at most one accepting path \cite[Deﬁnition I.1.11]{sakarovitch09}.
For trim automata (i) and (ii) are equivalent and one may be omitted.
Deterministic $K$-automata are unambiguous, but not every unambiguous $K$-automaton is equivalent to a deterministic one; furthermore not every $K$-automaton is equivalent to an unambiguous one.
Here, two $K$-automata are equivalent if they recognize the same $K$-rational series.

This leads to the following decidability problems for a $K$-automaton $\cA$.
\begin{itemize}
  \item \defit{Deterministic?} Is there a deterministic $K$-automaton $\cA'$ that is equivalent to $\cA$?
  \item \defit{Unambiguous?} Is there an unambiguous $K$-automaton $\cA'$ that is equivalent to $\cA$?
\end{itemize}
If these questions have a positive answer, it is furthermore desirable to actually produce a corresponding $K$-automaton.
These questions have received particular attention when $K$ is a tropical semiring \cite{choffrut77,mohri97,allauzen-mohri03,kirsten-maeurer05,kirsten-lombardy09,kirsten12,filiot-gentilini-raskin15,mohri-riley17}; the surveys \cite{lombardy-sakarovitch06,mohri09} are a good starting point.
Similar question have been studied for weighted tree automata \cite{buechse-vogler-may10,doerbrand-feller-stier21,fulop-koszo-vogler21,paul21}.
When $K$ is a field, even when $K=\bQ$, the question was still essentially completely open until recently.
It appears as an open problem in the 2006 survey by Lombardy and Sakarovitch \cite[Problem 1]{lombardy-sakarovitch06}.
For unary alphabets and $K=\bQ$, the problem \defit{Deterministic?} is decidable by a recent result of Kostolányi \cite{kostolanyi22}.
In the same setting \defit{Unambiguous?} is decidable by a result of Berstel and Mignotte \cite[Théorème 3]{berstel-mignotte76} together with a classical theorem of Pólya \cite[Chapter 6.3]{berstel-reutenauer11}.

In \cite{bell-smertnig21} a new invariant for an automaton with weights in a field, the \emph{linear hull}, was introduced, and it was used to prove a multivariate version of Pólya's theorem \cite[Theorem 1.2]{bell-smertnig21}.
This led to a characterization of $K$-rational series recognized by deterministic, respectively unambiguous, automata in terms of the linear hull of a minimal automaton for the series.
Unfortunately, the linear hull is defined as a topological closure (in the \emph{linear Zariski topology}) of the reachability set of an automaton, making its computability a non-trivial problem.

We show that the problems \defit{Deterministic?} and \defit{Unambiguous?} are decidable over number fields\footnote{The restriction to number fields is not essential, and only made for simplicity of the presentation.} (\cref{t:main-automata}).
Furthermore, our work yields an algorithm to compute an equivalent unambiguous, respectively, deterministic weighted automaton if it exists.
This uses the main theorems of \cite{bell-smertnig21} and a computability result for the linear hull (\cref{t:main}).

The key point is the computation of the linear Zariski closure of a matrix semigroup (a subsemigroup of the semigroup of all $d \times d$-matrices $M_d(K)$) generated by a closed set.
Our approach is inspired by the computation of the Zariski closure of such semigroups by Hrushovski, Ouaknine, Pouly, and Worrell \cite{hrushovski-ouaknine-pouly-worrell18}, which builds on the case for groups by Derksen, Jeandel, and Koiran \cite{derksen-jeandel-koiran05}; see also \cite{nosan-pouly-schmitz-shirmohammadi-worrell22}.
However, our approach stays almost entirely within the linear realm (see \cref{rem:efficiency}).

Our approach does not yield any bounds on the runtime.
The output size (the size of the linear hull) can be super-exponential in the input size. Namely, if $K=\bQ$ and $\cA$ has $d$ states, then the linear hull can be of size $2^{d-1}d!$ over a two-letter alphabet (\cref{r:output-size,rem:upper-bound-output}); by comparison, in the unary case, the algorithm of Kostolányi needs at most $O(d^3)$ operations.

In the group case (\cref{sec:invertible}), the Burnside--Schur theorem yields an upper bound on the size of a transversal modulo the component containing the identity, giving a bound on the output size that is double-exponential in $d$ (independent of the number of generators; \cref{rem:upper-bound-output}).
In the semigroup case (\cref{sec:semigroup}), this can be combined with a recursion lemma (\cref{l:recursion-newcpr}), to get a similar double-exponential upper bound (now dependent on the number of generators).
Further, our results  hold for all fields over which it is possible to do linear algebra exactly, and they can be extended to integers as well.
For reasons of space and simplicity we relegate details of this and the bounds on the output size to the arXiv version \cite{bell-smertnig23-arxiv}.

\textbf{Notation.} Throughout, let $K$ be a number field (a finite-dimensional field extension of $\bQ$), and let $d \ge 0$.
Let $M_d(K)$ be the semigroup of $d \times d$-matrices.
Further, $I \in M_d(K)$ denotes the identity matrix, and $E_{ij} \in M_d(K)$ denotes the $ij$-th elementary matrix.
If $X$ is a subset of a semigroup $\cS$, then $\langle X \rangle$ is the subsemigroup generated by $X$.
If $a$,~$b \in \bZ$, then $[a,b] \coloneqq \{\, x \in \bZ : a \le x \le b\,\}$ is the discrete interval.
Background on automata can be found in \cite{berstel-reutenauer11,droste-kuich-vogler09,sakarovitch09}.

\textbf{Acknowledgements.} We thank the reviewers for innumerable valuable comments on improving the presentation of the paper for the LICS community.
We have tried to implement them as far as possible; any remaining shortcomings are our own.

\section{Main results: Decidability of Deterministic? and Unambiguous?} \label{sec:main}

In this section we state the main results of the present paper (\cref{t:main,t:main-automata}) and show how \cref{t:main-automata} follows from \cref{t:main} and the results in \cite{bell-smertnig21}.
The proof of \cref{t:main} will then take up the rest of the paper.

We work with row vectors and apply matrices on the right.
A $d$-dimensional \defit{linear representation} over the alphabet $\Sigma$ consists of a row vector $u \in K^{1 \times d}$, a monoid homomorphism $\mu \colon \Sigma^* \to M_d(K)$, and a column vector $v \in K^d$.

To interpret $(u,\mu,v)$ as a $K$-automaton $\mathcal A$, we associate to it a directed graph with edge labels and set of vertices $[1,d]$ as follows: $u=(u_1,\ldots,u_d)$ is the vector of initial weights, with an incoming edge to state $i$ with weight $u_i$.
Analogously $v$ is interpreted as vector of terminal weights.
For each $a \in \Sigma$, the matrix $\mu(a)$ is an incidence matrix encoding the transition weights of the letter $a$: the $ij$-entry of $\mu(a)$ corresponds to the weight of the transition from state $i$ to the state $j$ labeled by $a$, and it is recorded by putting an edge with label $\mu(a) a$ (omitting the edge if $\mu(a)=0$).
In this way, there is a one-to-one correspondence between linear representations and weighted automata (see \cite[Chapter 1.6]{berstel-reutenauer11} for a more complete treatment).

An \emph{accepting path} for a word $w$ is a path in the graph that is labeled by $w$ and leads from an input state (a state with nonzero input weight) to a terminal state (a state with nonzero terminal weight).
We always assume that our automata are \emph{trim} (every state lies on some accepting path).

Given any word $w \in \Sigma^*$ one can compute the output $\mathcal A(w)\coloneqq u\mu(w)v$ of the $K$-automaton by
\begin{enumerate}
\item for each accepting path labeled by $w$, taking the product of the weights along each path;
\item summing up these values over all accepting paths for $w$. 
\end{enumerate}

The task of finding all accepting paths for $w$ becomes computationally easier if the automaton is
\begin{enumerate}
\item \emph{deterministic}, that is, there exists at most input state and for every state $i$ and every letter $a \in \Sigma$, there is at most one outgoing edge from $i$ that is labeled by $a$ (i.e., every row of $\mu(a)$ has at most one nonzero entry); or
\item \emph{unambiguous}, that is, for every word $w$ there exists at most one accepting path.
\end{enumerate}
Every deterministic automaton is unambiguous.

To an automaton we associate its \emph{behavior}, the $K$-rational series $\sum_{w \in \Sigma^*} \mathcal A(w) w$.
Two automata are equivalent if they have the same behavior.
Our main theorem is the following.

\begin{theorem} \label{t:main-automata}
  Let $\mathcal A$ be a $K$-automaton.
  Then it is decidable if $\cA$ is equivalent to
  \begin{enumerate}
  \item a deterministic $K$-automaton;
  \item an unambiguous $K$-automaton. 
  \end{enumerate}
  In both cases the corresponding deterministic (or unambiguous) $K$-automaton is computable.
\end{theorem}

To prove \cref{t:main-automata}, we will make use of the following linear version of the Zariski topology introduced in \cite[Section 3]{bell-smertnig21}.
The same topology previously appeared in work of Colcombet and Petrisan \cite{colcombet-petrisan17} under the name of ``glued spaces'' --- their minimal cover \cite[p.6]{colcombet-petrisan17} of a set of vectors is the closure of that set in the linear Zariski topology.

\begin{definition}
  On a finite-dimensional vector space $V$ over $K$, the \defit{linear Zariski topology} is the topology in which a set is closed if and only if it is a finite union of vector subspaces. 
\end{definition}

The empty set is represented by the empty union.
By definition, a (not necessarily closed) nonempty subset $X \subseteq V$ is \emph{irreducible}, if whenever $X \subseteq Y_1 \cup Y_2$ with closed sets $Y_1$ and $Y_2$, then already $X \subseteq Y_1$ or $X \subseteq Y_2$.
Since a vector space cannot be covered by finitely many proper subspaces (due to $K$ being infinite), one sees easily that the irreducible closed sets are precisely the vector subspaces of $V$, and every closed set can be expressed uniquely as the finite union of its irreducible components (i.e., the maximal irreducible subsets).\footnote{In fact, $V$ is a noetherian topological space, background on which can be found in \cite[\S II.4.1 and \S II.4.2]{bourbaki:ca72} or \cite[Sections \href{https://stacks.math.columbia.edu/tag/004U}{004U} and \href{https://stacks.math.columbia.edu/tag/0050}{0050}]{stacks-project}.}

Most of the paper is devoted to the following.
\begin{theorem} \label{t:main}
  Let $X \subseteq M_d(K)$ be a closed subset \textup(given by a list of basis vectors\textup) and let $\cS= \langle X \rangle$ be the semigroup generated by $X$.
  Then the linear Zariski closure $\overline{\cS}$ is computable \textup(as a list of basis vectors\textup).
\end{theorem}

\Cref{t:main} immediately yields the following corollary, by taking $X$ to be the union of the $n$ one-dimensional spaces generated by $A_1$, $\ldots\,$,~$A_n$.
\begin{corollary} \label{c:main-fg}
  Let $A_1$, $\ldots\,$,~$A_n \in M_d(K)$.
  Then the linear Zariski closure of the semigroup $\langle A_1, \ldots, A_n \rangle$ is computable.
\end{corollary}

We are now able to define the following crucial invariant of a weighted automaton over a field.

\begin{definition}
  Let $\cA$ be a $K$-automaton on the alphabet $\Sigma$ with linear representation $(u,\mu,v)$. The \defit{\textup{(}left\textup{)} linear hull} of $\cA$ is the set
  \[
    \overline{u \mu(\Sigma^*)} = \overline{\{\, u\mu(w) : w \in \Sigma^* \,\}},
  \]
  that is, it is the closure in the linear Zariski topology of the reachability set $\{\, u \mu(w) : w \in \Sigma^* \,\}$.
\end{definition}

The linear hulls of two equivalent $K$-automata need not coincide.
However, since $K$ is a field, there always exist minimal linear representations, and these are unique up to conjugation by an invertible matrix (corresponding to a change of basis of the vector space).
Correspondingly, the linear hulls of minimal linear representations only differ by a linear isomorphism on the ambient space.
In particular, the number of irreducible components and their dimensions are independent of the choice of minimal linear representation.
To a $K$-rational series we associate the linear hull of a minimal linear representation of the series.

The linear hull is \emph{not} left/right symmetric.
In fact the number of its irreducible components, on the left/right need not coincide, and neither need the dimensions \cite[Example 3.8]{bell-smertnig21}.

\begin{corollary} \label{c:lhull-compute}
  Let $\cA$ be a $K$-automaton.
  Then the linear hull of $\cA$ is computable.
\end{corollary}

\begin{proof}
  By \cref{c:main-fg} we can compute the linear Zariski closure of the finitely generated matrix semigroup $\mu(\Sigma^*)$.
  Since $\varphi\colon M_d(K) \to K^{1\times d},\ A \mapsto uA$ is $K$-linear, it is continuous in the linear Zariski topology and also closed (i.e., it maps closed sets to closed sets).
  Therefore $\overline{u\mu(\Sigma^*)} = u \overline{\mu(\Sigma^*)}$.
\end{proof}

Constructing the following automaton $\hat{\cA}$ is key in the decidability problem.

\noindent
\textbf{Construction of $\hat \cA$.} Given a $K$-automaton $\cA$, with minimal linear representation $(u,\mu,v)$, and linear hull $X = W_1 \cup \dots \cup W_k$ (where $W_1$, $\ldots\,$,~$W_k$ are irreducible components, with $m_i\coloneqq \dim W_i$), we can construct an equivalent $K$-automaton $\hat{\cA}$, with linear representation $(u',\mu',v')$, as follows (see \cite[Lemma 3.13]{bell-smertnig21} for a rigorous treatment): Renumbering the components, without restriction $u \in W_1$.
For each $a \in \Sigma$ and $i \in [1,k]$ there exists some $j \in [1,k]$ such that $W_i \mu(a) \subseteq W_j$.
Here, $j$ need not be unique, but for each $a$ we can choose a transition function $f_a\colon [1,k] \to [1,k]$ such that $W_i \mu(a) \subseteq W_{f_a(i)}$ for all $a \in \Sigma$ and $i \in [1,k]$.

Set $m = m_1 + \cdots + m_k \ge d$, so that $K^{1\times m} \cong W_1 \oplus \cdots \oplus W_k$ (typically $m > d$).
The linear representation $(u',\mu',v')$ will be constructed on this space.
Viewing $\mu(a)$ as linear endomorphisms on $K^{1 \times d}$, we can restrict $\mu(a)$ to $W_i$ to obtain linear maps $\mu(a)|_{W_i}\colon W_i \to W_{f_a(i)}$. Putting these linear endomorphisms all together, we get the endomorphism $\mu'(a)$ on $K^{1 \times m}$.
For $u'$ one puts $u$ into the $W_1$-component and zeroes everywhere else; $v'$ is constructed analogously to the $\mu(a)$ by viewing $v$ as linear functional $K^{1 \times d} \to K$.
By \cite[Lemma 3.13]{bell-smertnig21} this gives a $K$-automaton $\hat \cA$ equivalent to $\cA$.

By construction, the matrices $\mu'(a)$ have a $m_1 \times \cdots \times m_k$ block structure, with the property that every row of blocks contains at most one nonzero block.
We say that $(u',\mu',v')$ is \emph{semi-monomial} if, in addition, in every block of every $\mu'(a)$, each column has at most one nonzero entry and the analogous property holds for $v'$ (thinking of $v'$ as $k$ blocks of size $m_i \times 1$).
Clearly, whether $(u',\mu',v')$ is semi-monomial is decidable.

\begin{proof}[Proof of \cref{t:main-automata}]
  First, we compute a minimal linear representation $(u,\mu,v)$ of $\mathcal A$ \cite[p.41--42]{berstel-reutenauer11}, say of dimension $d$.
  Let $\Gamma \coloneqq \{\, u \mu(w) v : w \in \Sigma^* \,\}$ denote the set of all outputs of the automaton.
  Using \cite[Lemma 3.11]{bell-smertnig21} we can pick the minimal linear representation in such a way that $u \mu(\Sigma^*) \subseteq \Gamma^{1 \times d}$.

  Now compute the linear hull $X = \overline{u \mu(\Sigma^*)}$ (\cref{c:lhull-compute}).
  Let $W_1$, $\ldots\,$,~$W_k$ denote the irreducible components of $X$, with $\dim(W_i)=m_i$ and $m\coloneqq m_1+\cdots + m_k$.
  We now construct the linear representation $(u',\mu',v')$, of dimension $m \ge d$ and with associated automaton $\hat{\mathcal A}$, that recognizes the same series.
  Once we have $\hat{\mathcal A}$, we are able to resolve the decidability problem:
  \begin{itemize}
  \item $\mathcal A$ is equivalent to a deterministic automaton if and only if $X$ has dimension $\le 1$ (that is, $m_i=1$ for all $i)$ \cite[Theorem 1.3]{bell-smertnig21}.
    In this case $\hat{\mathcal A}$ is deterministic \cite[Proof of Proposition 3.14]{bell-smertnig21}.
  \item The proof of \cite[Proposition 5.3]{bell-smertnig21} implies that $\mathcal A$ is equivalent to an unambiguous automaton if and only if the specific automaton $\hat{\cA}$ is semi-monomial, and this can easily be checked. \qedhere
  \end{itemize}
\end{proof}

Taking $K=\bQ$, this solves Problem 1 in \cite{lombardy-sakarovitch06}.
It remains to establish \cref{t:main}.
One way to do so, is to first compute the Zariski closure using \cite{hrushovski-ouaknine-pouly-worrell18}, from which the linear Zariski closure can then be obtained (see \cite[Theorem 1]{lefaucheux-ouaknine-purser-worrell21}).

However, it seems unnecessarily complex to first compute the closure in the finer topology, both in principle as well as in terms of computational complexity.
We present an alternate approach that stays almost entirely within the realm of linear algebra.
In particular, it avoids the need of using Gröbner bases and of computing in extension fields.
We proceed in three steps, that successively build on each other: first we consider the problem for a single invertible matrix (\cref{sec:onematrix}), then for a closed set $X$ in which the invertible matrices are dense (essentially, the group case; \cref{sec:invertible}), and finally the case for general closed sets $X$ (the semigroup case; \cref{sec:semigroup}).

The linear algebraic approach can be expected to allow a more practical implementation (avoiding inefficient Gröbner bases). Unfortunately, at one point we need to leave to linear realm in an essential way (\cref{line:nonlinear} of \cref{alg:semigroupclosure}; see \cref{rem:efficiency}).
This appears to be the main obstacle to a more efficient implementation.

If one wishes to avoid computations in extension fields, while still using Gröbner bases, it would also be possible to use our computation for the single matrix case (\cref{sec:onematrix}) as ``subroutine'' in \cite{derksen-jeandel-koiran05,hrushovski-ouaknine-pouly-worrell18}. The output then lies between the Zariski closure and the linear Zariski closure, and \cite[Theorem 1]{lefaucheux-ouaknine-purser-worrell21} can be used to find the latter.

\begin{remark} \label{r:output-size}
  The linear hull can have super-exponentially many components in the dimension $d$, already in the case where the matrices form a group.
  The group of signed permutation matrices is a finite subgroup of $\GL_d(\bQ)$ of order $2^d d!$.
  By a result of Feit (\cite{feit96}; see also the introduction of \cite{berry-dubickas-elkies-poonen-smyth04} or \cite[\S6]{kuzmanovich-pavlichenkov02}), for large $d$, this order is maximal among all finite subgroups of $\GL_d(\bQ)$.
  Its linear Zariski closure consists of a union of $2^{d-1}d!$ vector spaces of dimension $1$ (a signed permutation and its negative always lie in the same vector space).
  Even worse, the group of signed permutation matrices is $2$-generated for all $d$, so that a better bound in terms of the number of generators of the group and the dimension is also also hopeless.
  Since the signed permutation matrices act faithfully on $(1,2,\ldots,d)$, this group also gives a linear hull of size $2^{d-1} d!$ for a two-letter alphabet and $d$ states.
\end{remark}

\section{A single invertible matrix} \label{sec:onematrix}

In this section, given $A \in \GL_d(K)$ we compute $\overline{\langle A \rangle}$.
Basic linear algebra, in particular generalized eigenspaces and the Jordan normal form, are sufficient to do so.
While computing the Jordan normal form usually involves computations in a finite extension of $K$ (for all the eigenvalues to be present), we get an algorithm that works over the initial field $K$.

We first need to understand the structure of the closure of a semigroup in the linear Zariski topology.
First note the following behavior of the closure with respect to products.

\begin{lemma} \label{l:closure-multiplication}
  Let $X \subseteq M_d(K)$ be a closed set, and let $D$,~$D' \subseteq X$ be arbitrary subsets.
  If $D D' \subseteq X$, then also $\overline{D} \, \overline{D'} \subseteq X$.
\end{lemma}

\begin{proof}
  Let $d' \in D'$.
  Then $Dd' \subseteq X$.
  Since multiplication by $d'$ from the right is linear, hence continuous and closed, also $\overline{D} d'=\overline{Dd'} \subseteq X$.
  Now we know $\overline{D} D' \subseteq X$, and still have to show $\overline{D}\,\overline{D'} \subseteq X$.
  Let $d \in \overline{D}$.
  From $dD' \subseteq X$ we find $d \overline{D'} = \overline{dD'} \subseteq X$.
  Thus $\overline{D}\,\overline{D'} \subseteq X$.
\end{proof}

\begin{lemma} \label{l:closure-structure}
  Let $\cS \subseteq M_d(K)$ be a subsemigroup.
  \begin{enumerate}
  \item \label{closure-structure:semigroup} The closure $\overline{\cS}$ is a semigroup.
  \item \label{closure-structure:group} If $\overline{\cS} \cap \GL_d(K) \ne \emptyset$, then $\overline{\cS} \cap \GL_d(K)$ is a linear algebraic group.
  \item \label{closure-structure:monoid} If $\cS \subseteq M_d(K)$ is a closed monoid \textup(a closed semigroup containing the identity matrix\textup), there exists a unique irreducible component $\cS^{0}$ containing the identity matrix.
    Then $\cS^{0}$ is a submonoid of $\cS$.
  \end{enumerate}
\end{lemma}

\begin{proof}
  \ref{closure-structure:semigroup}
  We have $\cS \cS \subseteq \cS \subseteq \overline{\cS}$.
  \cref{l:closure-multiplication} implies $\overline{\cS} \,\overline{\cS} \subseteq \overline{\cS}$.

  \ref{closure-structure:group}
  Clearly $\overline{\cS} \cap \GL_d(K)$ is a Zariski-closed subsemigroup of $\GL_d(K)$.
  Therefore it is a group \cite[Lemma 10]{derksen-jeandel-koiran05}.

  \ref{closure-structure:monoid} By \cite[Remark 5.2]{putcha88} (the proof is the same as the one for linear algebraic groups).
\end{proof}

Our main theorem in this section is the following.

\begin{theorem} \label{t:one-matrix}
  There exists a computable $N=N(d,K)$, such that for every $A \in \GL_d(K)$ we have $\overline{\langle A \rangle}^{0} = \lspan\{\, A^{Ni} \,:\, i \ge 0 \,\}$.
  In particular, $\overline{\langle A \rangle}$ is computable.
\end{theorem}

By $\mu(\QQbar)$ we denote the group of all roots of unity, where $\QQbar$ denotes the algebraic closure of $\bQ$, which is also the algebraic closure of $K$.

\begin{lemma} \label{l:invariant-no-roots-of-unity}
  Let $A \in \GL_d(K)$.
  Assume that for any two eigenvalues $\lambda$,~$\lambda' \in \QQbar$ of $A$ for which $\lambda/\lambda' \in \mu(\QQbar)$, it holds that $\lambda=\lambda'$.
  Let $n \ge 1$.
  Then a vector space $V \subseteq K^{d}$ is $A$-invariant if and only if it is $A^n$-invariant.
\end{lemma}

In the following proof we make use of the identity $a^n - b^n = \prod_{j=0}^{n-1} (a-\zeta^jb)$,
if $a$, $b$ commute and $\zeta$ is an $n$-th root of unity.

\begin{proof}[Proof of \cref{l:invariant-no-roots-of-unity}]
  If $V$ is $A$-invariant, then it is $A^n$-invariant.
  It suffices to show the converse.
  Without restriction we work over $\QQbar$.
  For every $\lambda \in \QQbar$, the space $V$ is $A$-invariant if and only if it is $(A - \lambda I)$-invariant.
  If $\lambda_1$, $\ldots\,$,~$\lambda_r$ are the pairwise distinct eigenvalues of $A$, then every generalized eigenspace $\ker(A-\lambda_i I)^d$ is $A$-invariant.
  If $V$ is $A$-invariant, we can consider the generalized eigenspaces of the restriction $A|_V$ to obtain a decomposition
  \[
     V =  \bigoplus_{i=1}^r (\ker(A- \lambda_i I)^d \cap V).
  \]

  Let $\lambda$ be an eigenvalue of $A$, and let $\zeta$ be a primitive $n$-th root of unity (which exists because $\QQbar$ is algebraically closed).
  Then
  \[
    \begin{split}
      (A^n - \lambda^n I)^i &= (A-\lambda I)^i \prod_{j=1}^{n-1} (A-\zeta^j \lambda I)^i \\
      &= \prod_{j=1}^{n-1} (A-\zeta^j \lambda I)^i \cdot (A-\lambda I)^i.
    \end{split}
  \]
  for $i \ge 0$.
  By our assumption on the ratios of eigenvalues, none of the $\zeta^j \lambda$ with $j \in [1,n-1]$ are eigenvalue of $A$.
  Thus, the matrices $(A-\zeta^j \lambda I)^i$ are invertible for $j \in [1,n-1]$.
  Consequently $\ker(A^n - \lambda^n I)^i = \ker(A-\lambda I)^i$.

  Let $\lambda_1$, $\ldots\,$,~$\lambda_r$ denote the pairwise distinct eigenvalues of $A$.
  Since $V$ is $A^n$-invariant,
  \[
    V = \bigoplus_{i=1}^r (\ker(A^n - \lambda_i^n I)^d \cap V) = \bigoplus_{i=1}^r (\ker(A-\lambda_i I)^d \cap V).
  \]
  It therefore suffices to show the claim when $A$ has a single eigenvalue $\lambda$.

  Since $V$ is $A^n$-invariant, it is also $(A^n - \lambda^n I)$-invariant.
  We show that it is $(A- \lambda I)$ invariant, then it is also $A$-invariant.
  It suffices to show that for every $0 \ne v \in V$ and all $i \ge 0$ we have $(A-\lambda I)^iv \in V$.

  Let $0 \ne v \in V$.
  For all $i \ge 0$, let $v_i \coloneqq (A-\lambda I)^iv$ and $v_i' \coloneqq (A^n-\lambda^n I)^iv$.
  Let $k \ge 0$ be minimal such that $v \in \ker((A-\lambda I)^{k+1}) = \ker((A^n -\lambda^n I)^{k+1})$.
  Then $v_k$ is an eigenvector of $A$ with respect to the eigenvalue $\lambda$.
  Thus
  \begin{align*}
    0 \ne v_k' &= \Big( \sum_{j=0}^{n-1} A^j \lambda^{n-1-j} \Big)^k (A-\lambda I)^k v  \\
     &= \Big( \sum_{j=0}^{n-1} A^j \lambda^{n-1-j} \Big)^k v_k = (n \lambda^{n-1})^k v_k.
  \end{align*}
  Hence $v_k' \in V$ implies $v_k \in V$.

  Suppose now that $v_{k}$, $\ldots\,$,~$v_{i+1} \in V$; we show $v_i \in V$.
  Again
  \[
    v_i' = \Big( \sum_{j=0}^{n-1} A^j \lambda^{n-1-j} \Big)^i (A-\lambda I)^i v =  \Big( \sum_{j=0}^{n-1} A^j \lambda^{n-1-j} \Big)^i v_i.
  \]
  Now $Av_i = \lambda v_{i+1}$, and so $A^j v_i \in \lspan\{ v_k,\ldots, v_{i+1} \} \subseteq V$ for all $j \in [1,n-1]$.
  Since also $v_i' \in V$, we get $v_i \in V$.
\end{proof}

\begin{lemma} \label{l:computable-n0}
  There exists a computable $N_0=N_0(d,K)$ such that, for every finite field extension $L/K$ with $[L:K] \le d$ and every root of unity $\zeta \in L$, one has $\zeta^{N_0}=1$.
\end{lemma}

\begin{proof}
  Let $\zeta \in L$ be a root of unity of some order $n \ge 1$.
  Then
  \[
    [L:\bQ] \ge [\bQ(\zeta):\bQ]=\phi(n),
  \]
  with $\phi(n)$ denoting the Euler-$\phi$-function.
  Since
  \[
    [L:\bQ] = [L:K][K:\bQ] \le d [K:\bQ],
  \]
  we must have $\phi(n) \le d [K:\bQ]$.
  Since $\phi(n) \to \infty$ as $n \to \infty$, but the right hand side of the inequality is constant, only finitely many values are possible for $n$.
  By taking $N_0$ to be the least common multiple of these values, the claim follows.
\end{proof}

The constant $N_0=N_0(d,K)$ in the previous lemma is explicit and does not depend on the matrix $A$.

\begin{lemma} \label{l:invariant-exponent}
  Let $N \coloneqq N(d,K) \coloneqq  N_0(d^2,K)$.
  Let $A \in \GL_d(K)$, and let $V \subseteq K^{d}$ be a vector subspace.
  If $V$ is $A^n$-invariant for some $n \ge 1$, then $V$ is $A^N$-invariant.
\end{lemma}

\begin{proof}
  Let $\lambda$,~$\lambda' \in \QQbar$ be eigenvalues of $A$ and let $N = N_0(d^2,K)$.
  Since $\lambda$,~$\lambda'$ are both roots of the characteristic polynomial, which has degree $d$, the extension $K(\lambda,\lambda')/K$ has degree at most $d^2$.
  If there exists a root of unity $\zeta$ such that $\lambda/\lambda' = \zeta$, then $\zeta \in K(\lambda,\lambda')$ and hence $\zeta^{N} = 1$.
  Thus $A^{N}$ satisfies the assumption of \cref{l:invariant-no-roots-of-unity}.

  Suppose now that $V$ is $A^n$-invariant ($n \ge 1$).
  Then $V$ is $A^{nN}$-invariant.
  \cref{l:invariant-no-roots-of-unity} gives that $V$ is $A^N$-invariant.
\end{proof}

Let $A \in \GL_d(K)$.
We recall (\cref{l:closure-structure}), that $\overline{\langle A \rangle} \cap GL_d(K)$ is a linear algebraic group, and $\overline{\langle A \rangle}$ has a unique irreducible component containing $I$.
This component is denoted by $\overline{\langle A \rangle}^0$.

\begin{proof}[Proof of \cref{t:one-matrix}]
  Let $Z_0\coloneqq \overline{\langle A \rangle}^0$.
  Since $A$ acts by permutation on the finitely many irreducible components of $\overline{\langle A\rangle}$, there exists an $N > 0$ such that $A^N Z_0 = Z_0$. 
  \Cref{l:invariant-exponent} implies that we can take $N = N(d,K)$, which is computable without knowing $Z_0$.
  
  Now $A^N \in Z_0$ and hence $\langle A^{N} \rangle \in Z_0$, because $Z_0$ is a submonoid of $\overline{\langle A \rangle}$ (by \ref{closure-structure:monoid} of \cref{l:closure-structure}).
  Since $Z_0$ is a vector space, even $\lspan \langle A^N \rangle  \subseteq Z_0$.
  Thus $\langle A \rangle \subseteq \bigcup_{i=0}^{N-1} A^i \lspan{\langle A^N \rangle} \subseteq \bigcup_{i=0}^{N-1} A^i Z_0 \subseteq \overline{\langle A \rangle}$.
  Taking closures, we get equality throughout, so $\overline{\langle A \rangle}^{0} = \lspan{\langle A^N \rangle}$.
  
  Finally, by the Cayley-Hamilton theorem there exist (computable) $\lambda_0$, $\ldots\,$,~$\lambda_{d-1} \in K$ such that $(A^N)^d + \lambda_{d-1} (A^N)^{d-1} + \cdots + \lambda_0 I = 0$.
  Multiplying by $A^{Nm}$ for $m \ge 0$, we see inductively that $\lspan{\langle A^N \rangle} = \lspan\{I, A^N, A^{2N}, \ldots, A^{(d-1)N}\} $.
  Thus $\overline{\langle A \rangle}^{0}$ and $\overline{\langle A \rangle}$ are computable.
\end{proof}

\begin{example}
  Let
  \[
    A =
    \begin{pmatrix}
      2 & 1 & 0 & 0 & 0 & 0 \\
      0 & 2 & 0 & 0 & 0 & 0 \\
      0 & 0 & 3 & 1 & 0 & 0 \\
      0 & 0 & 0 & 3 & 1 & 0 \\
      0 & 0 & 0 & 0 & 3 & 0 \\
      0 & 0 & 0 & 0 & 0 & -2
    \end{pmatrix}.
  \]
  Since $\phi(n) > 6^2=36$ for $n > 126$, we can take $N=N(6,\bQ)=126$.
  But since the only root of unity appearing for the specific $A$ is $-1$, we can actually take $N=2$.
  Setting $B=A^2$ we find $Z_0= \lspan\{\, I, B, B^2, B^3, B^4, B^5 \,\}$ to be $Z_0 = \lspan\{ E_{11}+E_{22}+E_{66}, E_{12}, E_{33}+E_{44}+E_{55}, E_{34}+E_{45}, E_{35} \}$.
  Finally $\overline{\langle A \rangle} = Z_0 \cup -E_{66} Z_0$.
  Up to base change the same is true for any matrix with Jordan normal form $A$.
\end{example}

\begin{remark}
  Instead of using the bound $N(d,K)$ one may compute the eigenvalues of $A$ explicitly in a suitable number field.
  It is then possible to compute the pairwise ratio of the eigenvalues and check which ones are a root of unity.
  This has the disadvantage of having to perform computations in a field extension of $K$ and that the resulting $N$ depends on $A$.
  However, the resulting $N$ could be much smaller than  $N(d,K)$.
\end{remark}

\section{Invertible Matrices} \label{sec:invertible}

In this section we consider the computation of $\overline{\langle X \rangle}$ when $X \subseteq M_d(K)$ is a closed set, and each irreducible component of $X$ contains invertible matrices.
In this case, $\overline{\langle X \rangle} \cap \GL_d(K)$ is a linear algebraic group (\cref{l:closure-structure}).
The algorithm is that of \cite{derksen-jeandel-koiran05}, with the Zariski topology replaced by the linear Zariski topology.
However, care must be taken in checking the correctness of the algorithm, as the use of the linear Zariski topology introduces some subtle difficulties.
We first state the algorithm, \cref{alg:groupclosure}, and illustrate it on a short example.

\begin{algorithm}
  \caption{\small Computation of  $\overline{\langle X \rangle}$ when the invertible matrices are dense in $X$.
    The irreducible components $Z_1$, $\ldots\,$,~$Z_r$ are given by their bases.
    Throughout the algorithm, $N$ is an irreducible closed set, containing the identity matrix, that is monotonically increasing with each iteration.
  Similarly, $T$ is a finite subset of $\langle I, A_1,\ldots,A_n \rangle$ that is monotonically increasing.}
   \label{alg:groupclosure}
  \begin{algorithmic}[1]
    \Function{GroupClosure}{$X$}
    \State $Z_1$, $\ldots\,$~$Z_l \gets \textrm{Irreducible components of $X$}$ 
    \Require{$\GL_d(K) \cap Z_i \ne \emptyset$ for all $i \in [1,l]$}
    \For{$i=1,\ldots,l$}
    \State $A_i \gets \textrm{An invertible element of $Z_i$}$ \label{alg1:choice}
    \EndFor
    \State $N \gets \overline{(A_1^{-1}Z_1) \cdots (A_l^{-1}Z_l)}$ \label{alg1:prod1} 
    \State $T \gets \{I,A_1,\ldots,A_l\}$
    \Repeat
    \State $N' \gets N$
    \State $T' \gets T$
    \For{$A \in T$}
    \State $N \gets \overline{N \overline{\langle A\rangle}^0}$ \label{alg1:0comp}
    \State $N \gets \overline{N(ANA^{-1})}$ \label{alg1:conj}
    \For{$B \in T$}
    \If{$AB \not\in T N$} \label{alg1:check-t}
    \State $T \gets T \cup \{AB\}$ \label{alg1:prod-tn}
    \EndIf
    \EndFor
    \EndFor
    \Until{$N'=N$ and $T'=T$}
    \State
    \textbf{return} $T N$
    \EndFunction
  \end{algorithmic}
\end{algorithm}

\begin{example}
  Consider
  \[
    A_1 =
    \begin{pmatrix}
      2 & 0 & 0 \\
      0 & 3 & 0 \\
      0 & 0 & -3 \\
    \end{pmatrix}, \quad
    A_2 =
    \begin{pmatrix}
      1 & 1 & 0 \\
      0 & 1 & 0 \\
      0 & 0 & 1
    \end{pmatrix},
  \]
  and $X = \bQ A_1 \cup \bQ A_2$.
  After initialization, $N=\{I\}$ and $T=\{I,A_1,A_2\}$.
  Now
  \begin{align*}
    \overline{\langle A_1 \rangle}^0 &= \lspan\{ E_{11}, E_{22}+E_{33}\},\\
    \overline{\langle A_2 \rangle}^0 &= \lspan\{ E_{11}+E_{22}+E_{33}, E_{12} \}.
  \end{align*}
  So $N$ becomes $\lspan\{E_{11},E_{22}+E_{33}, E_{12}\}$ in the first iteration of the loop at line 7 (\cref{l:prodclosure} below), and $TN=N \cup A_1N$ where $A_1N = \lspan\{E_{11},E_{22}-E_{33}, E_{12}\}$, so $T$ remains the same.
  In the second iteration $T$ and $N$ do not change anymore and the algorithm terminates.
\end{example}

In \cref{alg1:0comp} we make use of the case of a single invertible matrix to compute $\overline{\langle A \rangle}^0$.
Some steps need further elaboration:
\begin{enumerate}[label=\Alph*)]
\item In \cref{alg1:choice}, we need to be able to choose $A_i \in Z_i \cap \GL_d(K)$, under the assumption that this intersection is nonempty.
\item In \cref{alg1:prod1,alg1:0comp,alg1:conj}, we need to compute the closure of the product of two (or more) irreducible closed sets.
\end{enumerate}
We first explain these steps, and then show termination and correctness of the algorithm.

\subsection{Picking elements on which a polynomial does not vanish} \label{ss:cns}

The problem of picking an element in $Z_i \cap \GL_d(K)$ is an instance of the more general problem of picking an element in $Z_i$ on which a given polynomial (in this case, the determinant) does not vanish.
We give the general result, as we need it later.

Let $V \subseteq M_d(K)$ be a vector subspace and let $R=K[x_{ij} : 1 \le i,j \le d]$ be a polynomial ring in $d^2$ indeterminates.
Let $A_0 \in M_d(R)$ be the matrix whose $ij$-th entry is $x_{ij}$.
The space $V$ is defined by a finite number of homogeneous linear equations in the variables $x_{ij}$.
We can transform this system of equations into a triangular form by Gaussian elimination, and substitute into the entries of $A_0$ to eliminate a number of variables.
This leaves us with a matrix $A \in M_d(R)$ with the following property: Substituting any elements $\alpha_{ij} \in K$ for $x_{ij}$ yields a matrix in $V$, and conversely, every element of $V$ can be obtained in this way.
We call $A$ a \defit{generic matrix} of $V$.
\footnote{A more conceptual way to think about this is that the coordinate ring of $V$ is again a polynomial ring, and $A$ represents the homomorphism of coordinate rings $K[M_d(K)] \to K[V]$.}

\begin{lemma} \label{l:prod-containment}
  Let $V_1$, $\ldots\,$,~$V_n \subseteq M_d(K)$ be irreducible closed subsets.
  If $X \subseteq M_d(K)$ is a Zariski-closed subset, then it is possible to decide whether $V_1\cdots V_n \subseteq X$, and if this is not the case, to compute an element of $V_1\cdots V_n \setminus X$.
\end{lemma}

\begin{proof}
  Let $X$ be defined by nonzero polynomials $f_1$, $\ldots\,$,~$f_m \in K[x_{ij} : 1 \le i,j \le d]$.
  We may assume $m \ge 1$ as the claim is trivial otherwise.
  Represent each $V_k$ by a generic matrix $A_k \in M_d(K(\mathbf{y^{(k)}}))$, where $\mathbf{y^{(k)}}=(y_{ij}^{(k)})$ is a family of $d^2$ indeterminates.
  Then
  \begin{align*}
    V_1\cdots V_n = \{\, & A_1(\alpha^{(1)}_{ij}) \cdots A_n(\alpha^{(n)}_{ij})
    : \alpha^{(k)}_{ij} \in K,\,\\ & i,j \in [1,d],\, k \in [1,n] \,\}.
  \end{align*}
  Substituting, each of the polynomials $f_l(x_{ij})$ gives rise to a polynomial $g_l(y^{(1)}_{ij}, \ldots\, y^{(n)}_{ij}) \coloneqq f_l\big(A_1(y_{ij}^{(1)})\cdots A_n(y_{ij}^{(n)})\big)$ in at most $nd^2$ indeterminates.
  Now $V_1\cdots V_n \subseteq X$ if and only if all of $g_1$, $\ldots\,$,~$g_m$ vanish on $K^{nd^2}$.
  A polynomial $g_l$ ($l \in [1,m]$) vanishes on all of $K^{nd^2}$ if and only if it is the zero polynomial,\footnote{We use that $K$ is infinite.} and one checks this by simplifying the expression for $g_l$.

  Suppose now that some $g_l$ is nonzero. Let $\prod_{i,j,k} {(y_{ij}^{(k)})}^{t_{ij}^{(k)}}$ with $t_{ij}^{(k)} \ge 0$ be a monomial of maximal total degree in the support of $g_l$.
  Let $P_{ij}^{(k)} \subseteq K$ be a set of cardinality $t_{ij}^{(k)}+1$.
  By Alon's Combinatorial Nullstellensatz \cite[Theorem 1.2]{alon99}, the finite set
  \begin{align*}
    \{\, & g_l(\alpha_{ij}^{(1)},\ldots,\alpha_{ij}^{(n)}) = f_l\big(A_1(\alpha_{ij}^{(1)})\cdots A_n(\alpha_{ij}^{(n)})\big) : \\ & (\alpha_{ij}^{(k)}) \in M_d(K) \text{ with } \alpha_{ij}^{(k)} \in P_{ij}^{(k)} \,\}
  \end{align*}
  contains a nonzero element.
  Every such element gives rise to an element of $V_1\cdots V_k \setminus X$.
\end{proof}

\begin{example} \label{ex:cns}
  Let $V_1=\lspan\{ E_{11} + E_{12}, E_{21}+E_{22} \}$ and $V_2=\lspan\{ E_{11} + E_{21}, E_{12} + E_{22} \}$, with generic matrices
  \[
    A_1=
    \begin{pmatrix}
      x & x \\
      y & y \\
    \end{pmatrix},\quad
    A_2=
    \begin{pmatrix}
  z & w \\
  z & w \\
\end{pmatrix}.
\]
Set $f_1=x_{11}x_{22} - x_{12}x_{21}$ and $f_2 = (x_{11}-x_{21})(x_{11}-x_{12})$.
Evaluating $f_1$ and $f_2$ on the product of the generic matrices,
\[
  A_1A_2 = \begin{pmatrix}
  2xz & 2xw \\
  2yz & 2yw
\end{pmatrix},
\]
we get $g_1=4xzyw-4xwyz=0$ and $g_2=4(xz-yz)(xz-xw)$.
So $g_1$ vanishes on $V_1V_2$, but $g_2$ has a leading term $4yzxw$.
The Combinatorial Nullstellensatz implies that there is an element in $V_1V_2$ with $w$,~$x$,~$y$,~$z \in \{0,1\}$ on which $g_2$ does not vanish (e.g., $x=z=1$, $y=w=0$).
\end{example}

The special case of a single polynomial $f$ follows by setting $m=1$ and taking $X$ to be the vanishing set of $f$.

\subsection{Computing the closure of a product}

For vector subspaces $V$,~$W \subseteq M_d(K)$ we distinguish the pairwise product $VW \coloneqq \{\, vw : v \in V,\, w \in W \,\}$ which in general is not a vector space, and the product of vector spaces
\[
  V\bcdot W \coloneqq \lspan{VW} =  \lspan\{\, vw : v \in V,\, w \in W \,\},
\]
which is the span of the former.
We are interested mostly in closed sets, and the next lemma simplifies this issue.

\begin{lemma} \label{l:prodclosure}
  Let $V$,~$W \subseteq M_d(K)$ be irreducible closed subsets.
  Then $\overline{VW} = V\bcdot W$.
  In particular, the set $\overline{VW}$ is irreducible.
\end{lemma}

\begin{proof}
  The sets $V$, $W \subseteq M_d(K)$ are also closed and irreducible in the Zariski topology.\footnote{To see irreducibility, consider polynomials $f$,~$g \in K[x_{ij}]$ that vanish on proper subsets of $V$, and such that $fg$ vanishes on all of $V$.
  Using the linear homogeneous equations defining $V$, we can eliminate a number of variables in $f$ and $g$ to obtain nonzero polynomials $\hat f$, $\hat g$, in a subset of the variables $\{x_{ij}\}$, with the property that $\hat f \hat g$ vanishes everywhere. However, since $K$ is infinite, this implies $\hat f \hat g=0$, a contradiction to $\hat f$,~$\hat g \ne 0$.}
  In the Zariski topology the multiplication map $\mu \colon M_d(K) \times M_d(K) \to M_d(K), (A,B) \mapsto AB$ is continuous, and hence $\mu(V,W) = VW$ is irreducible \cite[\href{https://stacks.math.columbia.edu/tag/0379}{Lemma 0379}]{stacks-project}.
  Then $VW$ is also irreducible in the, coarser, linear Zariski topology.
  Thus the same is true for the closure $\overline{V W}$  \cite[\href{https://stacks.math.columbia.edu/tag/004W}{Lemma 004W}]{stacks-project}.
  So $\overline{V W}$ is a vector space.
  But $V\bcdot W$ is the smallest vector space containing $VW$, and thus $\overline{VW} = V\bcdot W$.
%
\end{proof}

Now it is easy to compute a generating set for $\overline{VW}$ as the pairwise products of bases of $V$ and $W$.

\begin{remark}
 The multiplication map $\mu$ is \emph{not} continuous in the linear Zariski topology.
 It is also possible to prove the previous lemma directly, without resorting to the Zariski topology, by showing $\overline{V W}=V \bcdot W$ by hand.
\end{remark}

\subsection{Termination and Correctness of \cref{alg:groupclosure}}

Recall that a group $G$ is a torsion group if every element has finite order. We need the following.

\begin{theorem}[Burnside--Schur {\cite[Theorem 2.3.5]{herstein94}}] \label{t:burnside-schur}
  If $G \le \GL_d(K)$ is a finitely generated torsion group, then $G$ is finite.
\end{theorem}

\begin{theorem} \label{t:alg1-correct}
  Let $X \subseteq M_d(K)$ be a closed subset \textup(given by a list of bases\textup) such that $\GL_d(K) \cap X$ is dense in $X$.
  Then $\overline{\langle X \rangle}$ is computable.
\end{theorem}

\begin{proof}
  We show that \cref{alg:groupclosure} terminates and yields $\overline{\langle X \rangle}$.
  The intersection $\GL_d(K) \cap \overline{\langle X \rangle}$ is a linear algebraic group by \ref{closure-structure:group} of \cref{l:closure-structure}, and we are going to use this structure.
  To do so, write $\zcl{X}$ for the closure of a set in the usual Zariski topology (i.e., not the linear one), taken over the algebraic closure $\QQbar$.
  
  Denote by $(T_1,N_1)$, $(T_2,N_2)$, $\ldots\,$, the subsequent values taken by $T$ and $N$.
  Then $N_1 \subseteq N_2 \subseteq \cdots$ is an ascending chain of vector subspaces of the finite-dimensional space $M_d(K)$, and $T_1 \subseteq T_2 \subseteq \cdots$ is an ascending chain of finite subsets of $\langle I, A_1, \cdots, A_l \rangle$.
  Define $N_\infty \coloneqq \bigcup_{i \ge 1} N_i$ and $T_\infty \coloneqq \bigcup_{i \ge 1} T_i$.
  Set $\cS = \bigcup_{i \ge 1} T_i N_\infty = T_\infty N_\infty$.\footnote{The idea will be that $N_\infty \cap \GL_d(K)$ is the irreducible component containing the identity, and $T_\infty$ is in fact a finite set that contains a transversal of the group $\overline{\langle X\rangle} \cap \GL_d(K)$ with respect to $N_\infty \cap \GL_d(K)$.}
  By construction $X$ is dense in $\cS$ (this is true in the beginning of the algorithm and is preserved in each step, keeping in mind \cref{l:closure-multiplication}).

  It remains to show that the algorithm terminates and that $\cS$ is a closed semigroup.
  Since each $N_i$ is a vector subspace of $M_d(K)$, the chain of $N_i$'s stabilizes at some $N_\infty = N_m$.
  For $i \ge 0$ and $A \in T_i$ note $AN_iA^{-1} \subseteq N_{i+1}$ (by line~\ref{alg1:conj}) and so $AN_\infty A^{-1} \subseteq N_\infty$.
  These being vector spaces of the same dimension, even $AN_\infty A^{-1} = N_\infty$.
  Let $H \coloneqq N_\infty \cap \GL_d(K)$.
  For every $i \ge 0$ and $A$,~$B \in T_i$ we have $(AH)(BH) \subseteq ABHH \subseteq T_{i+1}H$ by construction (the first inclusion by $BN_\infty=N_\infty B$; the second one by lines~\ref{alg1:conj} and \ref{alg1:prod-tn}).
  Therefore $G \coloneqq \bigcup_{i \ge 1} T_iH \subseteq \GL_d(K)$ is a semigroup.
  The Zariski closure $\zcl{G} \subseteq \GL_d(\QQbar)$ is a linear algebraic group, and $\zcl{H}$ is a closed normal subgroup.
  Indeed, as $N_\infty$ is a vector subspace of $M_d(K)$, the closure $\zcl{H}$ is simply the vector subspace of $M_d(\QQbar)$ defined by the same equations as $N_\infty$, intersected with $\GL_d(\QQbar)$.
  The quotient $\zcl{G}/\zcl{H}$ is also a linear algebraic group \cite[Theorem II.6.8]{borel91}, so without restriction $\zcl{G}/\zcl{H} \subseteq \GL_{d'}$ for some $d' \ge 1$.
  Let $\pi \colon \zcl{G} \to \zcl{G}/\zcl{H}$ denote the quotient morphism; it is a $K$-morphism of algebraic $K$-groups.

  By construction of the sets $T_i$ and $H$, the set $\pi(G)$ is contained in the subsemigroup of $\GL_{d'}(K)$ generated by $\pi(I)$, $\pi(A_1)$, $\ldots\,$,~$\pi(A_l)$.
  But it also contains all these elements, so $\pi(G) = \langle \pi(I), \pi(A_1), \ldots, \pi(A_l) \rangle$.
  By line~\ref{alg1:0comp}, every element of $\pi(G)$ has finite order.
  Therefore $\pi(G)$ is a torsion group.
  As we have just argued it is also finitely generated, and thus Burnside--Schur applies to show that $\pi(G)$ is finite.
  
  We now check that finiteness of $\pi(G)=\widetilde G/\widetilde H$ implies finiteness of $T_\infty$. Note that $\zcl{H} \cap \GL_d(K) = H$.
  Thus for $A$, $B \in \GL_{d}(K)$ we have $AB^{-1} \in \zcl{H}$ if and only if $AB^{-1} \in H$ if and only if $AB^{-1} \in N_\infty$.
  Looking at lines~\ref{alg1:check-t}--\ref{alg1:prod-tn}, once the chain $N_1 \subseteq N_2 \subseteq \cdots$ has stabilized at $N_\infty$, the chain $T_1 \subseteq T_2 \subseteq \ldots$ must also stabilize, say at the finite set $T_\infty= T_n$, because we are at this point only adding elements representing different cosets of $\widetilde G$ modulo $\widetilde H$.
  Then $\cS=T_nN_m$ is closed.

  Finally, $\cS$ is a semigroup: if $A$,~$B \in T_\infty$ then $(AN_\infty)(BN_\infty) = ABN_\infty N_\infty \subseteq ABN_\infty \subseteq T_\infty N_\infty$, where the last inclusion is ensured by line~\ref{alg1:prod-tn}.
\end{proof}

\section{Non-invertible matrices} \label{sec:semigroup}

Throughout this entire section, let $X \subseteq M_d(K)$ be a closed subset (in the linear Zariski topology) and let $\cS\coloneqq \langle X \rangle$ be the subsemigroup of $M_d(K)$ generated by $X$.
In this section we show how to compute the closure $\overline{\cS}$.

Since $\overline{\cS}$ is closed in the linear Zariski topology, it is also closed in the Zariski topology.
The set $\overline{\cS}$ is therefore a linear semigroup (\cref{l:closure-structure}) and in particular strongly $\pi$-regular (every element has a power that is contained in a subgroup of $\overline{\cS}$).
Much is known about the structure of linear semigroups \cite{putcha88,renner05}, respectively strongly $\pi$-regular matrix semigroups \cite[Section 2.3.2]{renner05} \cite{okninski85}.
These structural results are reflected in the algorithmic considerations, although they are not directly applicable to $\cS$ itself.
More general structural results about matrix semigroups, applying also to $\cS$, can be found in \cite{okninski91,okninski98}.
However, we will not be making use of them.

Our approach leans heavily on an algorithm for the computation of the Zariski closure, described in \cite{hrushovski-ouaknine-pouly-worrell18}.
However, we use more semigroup-theoretic language.
A key point in \cite{hrushovski-ouaknine-pouly-worrell18} is the use of an inductive approach based on the rank: first the closure of the semigroup generated by elements of the maximal rank $r$ is computed, then the closure of all elements of rank $\ge r-1$, and so on.

\begin{definition}
  For $\emptyset \ne X \subseteq M_d(K)$ closed, the \defit{generic rank} of $X$ is $\grank(X) \coloneqq \max\{\, \rank(A) : A \in X \,\}$.
\end{definition}

A disadvantage arising from the coarseness of the linear Zariski topology compared to the Zariski topology is that the generic rank is ill-behaved with respect to products.

\begin{example} \label{exm:generic-rank}
  Consider again \cref{ex:cns}.
  Then $V_1$ and $V_2$ are $2$-dimensional vector spaces of generic rank $1$.
  However, $V_1V_2$ contains all matrices $E_{ij}$.
  Thus $\overline{V_1V_2}=M_2(K)$ has generic rank $2$.
  (In the usual Zariski topology, $V_1V_2$ is not dense in $M_2(K)$: the determinant vanishes on the entire set.)
\end{example}

\Cref{exm:generic-rank} shows that taking a closure of a product of vector spaces may introduce elements of larger rank.
Much of the difficulty in the linear Zariski topology setting revolves around ensuring termination in light of this ill-behaved nature of the generic rank (\cref{rem:efficiency}).

We call a matrix $A \in M_d(K)$ \defit{completely pseudo-regular} if it is contained in a subgroup of $M_d(K)$.%
\footnote{In semigroup theory, an element of $\cS$ is \defit{completely regular} if it is contained in a subgroup of $\cS$. Completely regular elements of $\cS$ are completely pseudo-regular, but the converse may fail if $\cS$ is not strongly $\pi$-regular: e.g., an inverse to a given matrix $A \in \cS$ may exist in $M_d(K)$ but not be contained in $\cS$.}
The main issue in computing $\overline{\cS}$, is that completely pseudo-regular elements $A$ of $\cS$ give rise to subgroups of $\overline{\cS}$, i.e., subsemigroups of $M_d(K)$ that are groups with regards to some idempotent matrix as identity.
We write $E(A)$ for this idempotent.
We will need to deal with these subgroups by reducing to the (already proven) group case.

For a subset $Y \subseteq M_d(K)$ and $n \ge 1$, we define
\[
  Y^{\le n} \coloneqq \bigcup_{k=1}^n Y^k = \{\, A_1 \cdots A_k  : k \in [1,n], A_1, \ldots, A_k \in Y \,\}
\]
and $Y^{\trianglelefteq n} \coloneqq \{I\} \cup Y^{\le n}$.
For a completely pseudo-regular element $A \in \cS$ of rank $r$ and a closed set $Y$, let $E=E(A)$, \label{def:tya}
\[
  \begin{split}
    \cT_0(Y,A)&\coloneqq \Big\{\, B \in \overline{E Y^{\trianglelefteq 2\binom{d}{r} + 5} E} : \rank(B) = r \,\Big\}, \ \text{ and }\ \\
    \cT(Y,A) &\coloneqq \overline{Y^{\trianglelefteq\binom{d}{r}+2}\ \overline{ \langle \cT_0(Y,A) \rangle}\ Y^{\trianglelefteq\binom{d}{r}+2}}.
  \end{split}
\]

We now have all the tools to state \cref{alg:semigroupclosure}.

\begin{algorithm} 
  \caption{\small Computation of  $\overline{\langle X \rangle}$ in the general case. \textsc{FindCPR} discovers a new completely pseudo-regular element of rank $>s$. \textsc{TryClose} returns a closed set, that is equal to $\overline{\langle X \rangle}$ if all necessary completely pseudo-regular elements have been discovered.}
   \label{alg:semigroupclosure}
  \begin{algorithmic}[1]
    \Function{SemigroupClosure}{$X$}
    \Require{$X$ closed set}
    \State $r \gets \grank(X)$, $R_1$, $\ldots\,$,~$R_r \gets \emptyset$
    \State $Y_i$, $T_i$ ($i \in [1,r]$) $\gets \textsc{TryClose}(X, R_1, \ldots, R_r)$
    \While{$\overline{Y_1^2} \not\subseteq Y_1$} \label{line:monoidloop}
    \State $s \gets 0$
    \Repeat \label{line:resetloop}
    \State $B \gets \textsc{FindCPR}(X, Y_i, T_i, R_i, s)$\label{line:cpr-call}
    \State $s \gets \rank(B)$,
    \State $R_s \gets R_s \cup \{B\}$; $R_1$, $\ldots\,$,~$R_{s-1} \gets \emptyset$
    \Until{$\card{R_s} \le \binom{d}{s}$}
    \State $Y_i$, $T_i$ ($i \in [1,r]$) $\gets \textsc{TryClose}(X, R_1, \ldots, R_r)$
    \EndWhile
    \textbf{return} $Y_1$
    \EndFunction
    \medskip
    \Function{TryClose}{$X$, $R_1$, $\ldots\,$,~$R_r$}
    \Require{$X$ closed set; $R_s \subseteq \langle X \rangle$ finite set of completely pseudo-regular elements of rank $s$}
    \State \label{line:ys-init}$r \gets \grank(X)$, $Y_{r+1} \gets X$
    \For{$s=r$, $\ldots\,$,~$1$}
    \State $T_s \gets \bigcup_{B \in R_s} \cT(Y_{s+1}, B)$\label{line:compute-class}
    \State\label{line:def-ys}$Y_s \gets \overline{(Y_{s+1} \cup T_s)^{\le 2\binom{d}{s}+3}}$
    \EndFor
    \textbf{return} $Y_1$, $T_1$, $\ldots\,$,~$Y_r$, $T_r$
    \EndFunction
    \medskip
    \Function{FindCPR}{$X$, $Y_1$, $T_1$, $R_1$, $\ldots\,$,~$Y_r$, $T_r$, $R_r$, $s$}
    \State $r \gets \grank(X)$
    \For{$n \ge 0$}
    \For{$s'=r$, $\ldots\,$, $s+1$}
    \State $C_{s'} \gets Y_{s'} \cup \{\, A \in M_d(K) : \rank(A) < s' \,\}$
    \If{$n \ge 2\binom{d}{s'}+4$ \text{and} $X^n \setminus C_{s'} \ne \emptyset$} \label{line:nonlinear}
    \State\label{line:pickelt}$A_1 \cdots A_n \gets$ \text{an element of $X^n \setminus C_{s'}$}
    \State\label{line:pickcpr}$A_k \cdots A_l \gets$ \text{c.p.r. subprod. $\not \in Y_{s'+1} \cup T_{s'}$}
    \State \textbf{return} \text{$A_k\cdots A_l$}
    \EndIf
    \EndFor
    \EndFor
    \EndFunction
  \end{algorithmic}
\end{algorithm}

The main idea in the algorithm is: each set $R_s$ is a (finite) set of completely pseudo-regular elements of rank $s$.
Under the assumption that each $R_s$ is actually a full set of representatives of completely pseudo-regular elements of rank $s$, we attempt to compute $\overline{\cS}$ using a recursive strategy (\textsc{TryClose}).
If this fails to yield the entire closure, then in fact some completely pseudo-regular element must be missing and we can find such an element (using \textsc{FindCPR}), add it to $R_s$, and try again.
We give an example illustrating the algorithm; afterwards we deal with the computation of $\cT(Y,A)$ (\cref{line:compute-class}) and termination and correctness of \cref{alg:semigroupclosure}.

\begin{example}
  Let
  \[
    A =
    \begin{pmatrix}
      2 & 0 & 0 \\
      0 & -2 & 0 \\
      0 & 0 & 3
    \end{pmatrix},
    B =
    \begin{pmatrix}
      0 & 0 & 0 \\
      0 & 0 & 1 \\
      1 & 1 & 0
    \end{pmatrix},
    C
    =
    \begin{pmatrix}
      0 & 0 & 0 \\
      0 & 0 & 0 \\
      0 & 0 & 5
    \end{pmatrix},
  \]
  and $X=A\bQ \cup B \bQ \cup C \bQ$. On the first iteration, in \textsc{TryClose}, all $R_s=T_s=\emptyset$ and $Y_4=X$, $Y_3=X^{\le 5}$ consists of all scalar multiples of nonempty products of at most $5$ of the matrices, $Y_2=(X^{\le 5})^{\le 9} = X^{\le 45}$, and $Y_1 =X^{\le 405}$. Now $\overline{Y_1^2} \not \subseteq Y_1$ (e.g., $A^{406}$ is not contained in $Y_1$).
  So the check on \cref{line:monoidloop} fails.

  Now \textsc{FindCPR} gets called (with $s=0$).
  It discovers $A^6 \in X^6 \setminus Y_3$, which, being invertible, is actually completely pseudo-regular with $E(A)=I$. However, to make the example more illustrative, we deviate here from the actual pseudo-code and presume that \textsc{FindCPR} would instead return the completely pseudo-regular element $C$.\footnote{Otherwise, the next call to \textsc{TryClose} already returns the entire closure, as we will see below.} Then $E(C)=E_{33}$ and $R_1=\{C\}$.

  One gets $T_1 = \cT(Y_2,C) = \lspan\{ E_{33} \} \cup \lspan \{ E_{23} \} \cup \lspan\{ E_{31}+E_{32} \} \cup \lspan\{ E_{31} - E_{32} \} \cup \lspan\{ E_{21} + E_{22} \} \cup \lspan\{ E_{21}  - E_{22} \}$ (note $T_1^2$, $XT_1$, $T_1X \subseteq T_1$).
  So, the second iteration of the loop at \cref{line:monoidloop} yields $Y_4=X$, $Y_3=X^{\le 5}$, $Y_2=X^{\le 45}$, and $Y_1 = (X^{\le 45} \cup T_1)^{\le 9} = X^{\le 405} \cup T_1$.
  However, again $\overline{Y_1^2} \not\subseteq Y_1$ and \textsc{FindCPR} gets called again.
  Let us assume that at this point \textsc{FindCPR} returns (correctly) $A^6$ (with $E(A^6)=I$).
  Then $R_3=\{A^6\}$, while now $R_2=R_1=\emptyset$ are reset.

  On the next call to \textsc{TryClose}, we get $T_3=\cT(A, I)=\lspan\{ E_{11}+E_{22}, E_{33}\} \cup \lspan\{ E_{11} - E_{22}, E_{33} \}$. Then $Y_4 = X$. Multiplying $T_3$ from the left by $B$, $B^2$, $C$, $CB$, $CB^2$, one can find
  \[
    \begin{split}
      Y_3 &= \lspan\{ E_{11}+E_{22}, E_{33}\} \cup \lspan\{ E_{11} - E_{22}, E_{33} \} \\
      &\cup \lspan\{ E_{21}+E_{22}, E_{33} \} \cup \lspan\{ E_{21}-E_{22}, E_{33} \}\\
      &\cup \lspan\{ E_{31}+E_{32}, E_{23} \} \cup \lspan\{ E_{31}-E_{32}, E_{23} \}\\
      &\cup \lspan\{ E_{33} \} \cup \lspan \{ E_{23} \} \\
      &\cup \lspan\{ E_{31}+E_{32} \} \cup \lspan\{ E_{31} - E_{32} \} \\
      &\cup \lspan\{ E_{21} + E_{22} \} \cup \lspan\{ E_{21}  - E_{22} \}.
    \end{split}
  \]
  Now one can check $Y_3^2 \subseteq Y_3$, so $Y_1=Y_2=Y_3$, and this is the closure of $\langle X \rangle$.

  Finally, if we multiply this set with $(1,1,1)=e_1+e_2+e_3$ from the left (i.e., summing the rows), we get
  \[
    (1,1,1)Y_3 = \lspan\{ e_1+e_2, e_3 \} \cup \lspan\{ e_1-e_2, e_3 \}.
  \]
  This is the linear hull of the automaton in \cite[Example 3.7]{bell-smertnig21}.
\end{example}

Before we can discuss correctness and termination of the algorithm, we show that the generic rank is computable (\cref{c:compute-generic-rank}), that $\cT(Y,A)$ is computable (\cref{l:subgroup-closure}), and that we need to consider only finitely many completely pseudo-regular elements (\cref{l:scc}), up to a certain equivalence (\cref{d:cpr-equivalence}).

\subsection{Computability of the generic rank}

To compute the generic rank, we relate it to generic matrices (\cref{ss:cns}).

\begin{lemma} \label{l:grank}
  Let $r \in \bZ_{\ge 0}$.
  For an irreducible closed subset $V \subseteq M_d(K)$, the following statements are equivalent.
  \begin{enumerate}[label=\alph*)]
  \item\label{grank:grank} $\grank(V) = r$.
  \item\label{grank:generic} Every generic matrix of $V$ has rank $r$.
  \item\label{grank:generic-one} There exists a generic matrix of $V$ with rank $r$.
  \item\label{grank:zdense} There exists a Zariski-dense Zariski-open subset $U \subseteq V$ with $\rank(A)=r$ for all $A \in U$.
  \end{enumerate}
\end{lemma}

\begin{proof}
  The implications \ref{grank:generic}$\,\Rightarrow\,$\ref{grank:generic-one} and \ref{grank:zdense}$\,\Rightarrow\,$\ref{grank:grank} are immediate from the definitions.

  \ref{grank:generic-one}$\,\Rightarrow\,$\ref{grank:zdense}
  Let $R = K[x_{ij} : 1 \le i,j \le d]$ and let $A \in M_d(R)$ be a generic matrix of $V$.
  Performing Gaussian elimination over the field of fractions $\quo{R}=K(x_{ij} : 1 \le i,j \le d)$ of $R$, we find an invertible matrix $T \in M_d(\quo{R})$ such that $B=TA$ is in reduced row echelon form.
  Let $f \in R$ be a nonzero common multiple of the denominators of the entries of $T$, $T^{-1}$, and $B$.
  Whenever $A(\alpha_{ij}) \in V$  with $f(A(\alpha_{ij})) \ne 0$, we get that $A(\alpha_{ij}) = T^{-1}(\alpha_{ij}) B(\alpha_{ij}) \in V$ is well-defined and has rank $r$ (as $B(\alpha_{ij})$ is still in reduced row echelon form and $T(\alpha_{ij})$ is invertible).
  The set $D(f) = \{\, A(\alpha_{ij}) \in V : f(A(\alpha_{ij})) \ne 0 \,\}$ is nonempty and Zariski-open in $V$.
  By Zariski-irreducibility of $V$ it is Zariski-dense in $V$.

  \ref{grank:grank}$\,\Rightarrow\,$\ref{grank:generic}
  Let $A$ be a generic matrix of $V$ with $\rank(A) = s$.
  In light of \ref{grank:generic-one}$\,\Rightarrow\,$\ref{grank:zdense} we see that $V$ contains a Zariski-dense subset $U$ of rank $s$ matrices.
  Thus $\grank(V) \ge s$.
  On the other hand, all $(s+1)\times (s+1)$ minors vanish on $U$.
  Since these minors are polynomials in the entries of the matrices, also all elements of the Zariski closure of $U$ have rank $\le s$.
  Altogether $\grank(V) = s$.
\end{proof}

The generic rank $r=\grank(V)$ can therefore be computed using Gauss elimination on a generic matrix of $V$.

\begin{corollary} \label{c:compute-generic-rank}
  Let $V \subseteq M_d(K)$ be an irreducible closed subset.
  Then $r=\grank(V)$ is computable.
\end{corollary}

\subsection{A key finiteness result}

The following will be applied in various guises.
(This observation has also been used in \cite{hrushovski-ouaknine-pouly-worrell18}. Similar considerations are used to derive the bounds in \cite{okninski85}.)

\begin{lemma} \label{l:key-finiteness}
  Let $W$ be a $d$-dimensional vector space.
  Let $r \in [0,d]$ and let $(U_1,V_1)$, $\ldots\,$,~$(U_n,V_n)$ be pairs of vector subspaces of $W$ such that $U_i \cap V_i = 0$ and $\dim V_i=r$ for $i \in [1,n]$.
  If $n > \binom{d}{r}$, then
  \begin{enumerate}
  \item there exist $i > j$ such that $U_i \cap V_j = 0$, and
  \item there exist $i < j$ such that $U_i \cap V_j = 0$.
  \end{enumerate}
\end{lemma}

\begin{proof}
  Replacing the $U_i$ by larger spaces if necessary we may suppose $\dim U_i=d-r$ for $i \in [1,n]$.
  Therefore it suffices to show the first claim, the second one follows by symmetry.

  Fixing bases $u_{i,1}$, $\ldots\,$,~$u_{i,d-r}$ of $U_i$ and $v_{i,1}$, $\ldots\,$,~$v_{i,r}$ of $V_i$ we can associate to $U_i$ and $V_i$ the elements $\alpha_i \coloneqq u_{i,1} \wedge \cdots \wedge u_{i,d-r} \in \bigwedge^{d-r} W$ and $\beta_i \coloneqq v_{i,1}\wedge \cdots \wedge v_{i,r} \in \bigwedge^{r} W$. (A different choice of bases only changes the corresponding $\alpha_i$, respectively, $\beta_i$ by a nonzero scalar multiple.)
  Now $U_i \cap V_j = 0$ if and only if $\alpha_i \wedge \beta_j \ne 0$ in the exterior algebra $\bigwedge W$.

  Assume, for the sake of contradiction, $U_i \cap V_j \ne 0$ for all $i$,~$j \in [1,n]$ with $i > j$.
  Then $\alpha_i \wedge \beta_j = 0$ for $i > j$ but $\alpha_i \wedge \beta_i \ne 0$.
  Thus $\beta_i$ cannot be a linear combination of $\beta_1$, $\ldots\,$,~$\beta_{i-1}$.
  Hence the $\beta_1$, $\ldots\,$,~$\beta_n$ are linearly independent in $\bigwedge^r W$, and therefore $n \le \dim \bigwedge^r W = \binom{d}{r}$ contradicts the assumption on $n$.
\end{proof}

\subsection{Equivalence classes of completely pseudo-regular elements.}

We need an intrinsic characterization of completely pseudo-regular elements.

\begin{lemma} \label{l:cr}
  Let $A \in M_d(K)$. The following statements are equivalent.
  \begin{enumerate}[label=\alph*)]
  \item \label{cr:cr} $A$ is completely pseudo-regular.
  \item \label{cr:inv} There exists $A' \in M_d(K)$ such that $A=AA'A$ and $AA'=A'A$.
  \item \label{cr:idemp} There exist $E$,~$A' \in M_d(K)$ such that $E^2=E$, $EA=AE=A$, and $AA'=A'A=E$.
  \item \label{cr:square} $\rank A = \rank A^2$.
  \item \label{cr:intersection} $\im(A) \cap \ker(A) = 0$.
  \end{enumerate}
\end{lemma}

\begin{proof}
  The equivalence of \ref{cr:cr}, \ref{cr:inv}, and \ref{cr:idemp} holds in all semigroups. For convenience, we recall a proof.

  \ref{cr:cr}$\,\Rightarrow\,$\ref{cr:inv} Let $G \subseteq M_d(K)$ be a subgroup containing $A$, and $A'$ the inverse of $A$ in $G$.
  
  \ref{cr:inv}$\,\Rightarrow\,$\ref{cr:idemp} $E\coloneqq A'A$ is idempotent as claimed.
  
  \ref{cr:idemp}$\,\Rightarrow\,$\ref{cr:cr} The semigroup generated by $A$, $A'$, and $E$ is a group.

  \ref{cr:inv}$\,\Rightarrow\,$\ref{cr:square} Since $A=A^2A'$, we have $\im(A) \subseteq \im(A^2)$, and hence $\im(A)= \im(A^2)$.
  
  \ref{cr:square}$\,\Leftrightarrow\,$\ref{cr:intersection} Clear.
  
  \ref{cr:intersection}$\,\Rightarrow\,$\ref{cr:idemp} We have $K^d = \im(A) \oplus \ker(A)$, and therefore it is possible to construct a suitable inverse to $A$ on $\im(A)$ and extend it to $K^d$.
\end{proof}

Suppose that $A$ is completely pseudo-regular and $E$ is an idempotent as in \ref{cr:idemp}.
Then $\rank A = \rank E$.
From this rank equality and $EA=A$ and $AE=E$, one deduces $\im E=\im A$ and $\ker E=\ker A$, so that $E$ is uniquely determined by $A$ (an idempotent matrix $E$ is a projection onto the subspace $\im E$ along $\ker E$, and it is therefore uniquely determined by its image and its kernel).
Then $E=E(A)$ is the identity element of any subgroup containing $A$.

The element $A'$ with $AA'=A'A=E$ is not uniquely determined, but there is a unique such $A'$ with $A' \in EM_d(K)E$ (because $A'|_{\im E}$ is determined by $A$ and $A'|_{\ker E}=0$). We write $A^+$ for this element of $EM_d(K)E$ and call it the \defit{pseudo-inverse} of $A$.

\begin{lemma} \label{l:closed}
  If $\cS \subseteq M_d(K)$ is a Zariski-closed subsemigroup, then $\cS$ is strongly $\pi$-regular. 
  For every completely pseudo-regular $A \in \cS$, also $E(A)$, $A^+ \in \cS$.
\end{lemma}

\begin{proof}
  A Zariski-closed semigroup $\cS$ is strongly $\pi$-regular by \cite[Theorem 3.18]{putcha88} and the remaining claims follow from inspection of the proof of the cited theorem.
\end{proof}

There may be infinitely many completely pseudo-regular elements (and associated subgroups), and we need to reduce the problem to one where we only have to deal with finitely many.
To do so, we deal with equivalence classes of completely pseudo-regular elements.

\begin{definition} \label{d:cpr-equivalence}
  \begin{enumerate}
  \item  For $A$,~$B \in M_d(K)$ write $A \parallel B$ if $\im(A)=\im(B)$ and $\ker(A)=\ker(B)$.
  \item  For $A$,~$B \in \cS$ let $A \sim_{\cS} B$ if there exist $C$,~$D$,~$C'$,~$D' \in \cS \cup \{I\}$ such that $B \parallel DAC$ and $A \parallel D'BC'$.
  \end{enumerate}
\end{definition}

The relation $\sim_{\cS}$ is an equivalence relation on $\cS$.
The rank is constant on each $\sim_\cS$-equivalence class, and we may therefore speak of the \defit{rank} of an equivalence class.
We write $[A]_\cS$ for the $\sim_\cS$-equivalence class of $A \in \cS$.

The rest of the subsection is dedicated to ultimately proving that, given a completely pseudo-regular element $A \in \cS$ of rank $r$, and under the assumption that we are able to compute a closed set $Y$ containing all elements of $\cS$ of rank $>r$, it is possible to compute a closed set that contains the entire equivalence class $[A]_\cS$ (this is \ref{subgroup-closure:J} of \cref{l:subgroup-closure}).
This will allow us to compute $\cT(Y,A)$.

The following lemma replaces \cite[Propositions 9 and 10]{hrushovski-ouaknine-pouly-worrell18} in our setting.

\begin{lemma} \label{l:subproduct}
  Let $A=A_1 \cdots A_n$ with $A_1$, $\ldots\,$,~$A_n \in M_d(K)$.
  Suppose there exists $r \ge 0$ such that $\rank(A)=\rank(A_iA_{i+1})=r$ for all $i \in [1,n-1]$.
  \begin{enumerate}
  \item \label{subproduct:parallel} There exists a subproduct $A'\coloneqq A_{i_1}\cdots A_{i_k}$ with $1=i_1 < i_2 < \cdots < i_{k-1} < i_k=n$ such that $A \parallel A'$ and $k \le \binom{d}{r}+3$.
  \item \label{subproduct:cpr} If $n \ge 2\binom{d}{r} + 4$, then there are $1 \le k < l \le n$ such that $A_k\cdots A_l$ is completely pseudo-regular of rank $r$.
  \end{enumerate}
\end{lemma}

\begin{proof}
  \ref{subproduct:parallel} For $i \in [3,n-1]$ define $V_i \coloneqq \im(A_i\cdots A_n)$ and $U_i\coloneqq \ker(A_1\cdots A_{i-1})$.
  Then $V_i \cap U_i = 0$ for all $i \in [3,n-1]$.
  Suppose $n > \binom{d}{r}+ 3$.
  By \cref{l:key-finiteness}, there exist $i$,~$j \in [3,n-1]$ with $j < i$ such that $U_{j} \cap V_i = 0$.
  Then
  \[
    A_1 \cdots A_{j-1}(A_j\cdots A_{i-1})A_i\cdots A_n \parallel A_1 \cdots A_{j-1}A_i\cdots A_n,
  \]
  and the second product has fewer factors.
  The claim follows by repeating this process.

  \ref{subproduct:cpr} For $i \in [1,\lfloor n/2 \rfloor -1]$, let $U_i = \ker(A_{2i-1}A_{2i})$ and $V_i = \im(A_{2i+1}A_{2i+2})$.
  Then $U_i \cap V_i = 0$ for all $i$.
  By \cref{l:key-finiteness}, there are $i < j$ with $U_j \cap V_i = 0$.
  Then $\im(A_{2i+1}A_{2i+2}) \cap \ker(A_{2j-1}A_{2j}) = 0$ and $2i+1 < 2j$, so $k=2i+1$ and $l=2j$ works.
\end{proof}

\begin{lemma} \label{l:conn}
  Let $A \in \cS$ be completely pseudo-regular and $B \in [A]_\cS$.
  \begin{enumerate}
  \item \label{conn:element}There exist completely pseudo-regular $C$,~$D \in [A]_\cS$ such that $B=E(D)B$ and $B=BE(C)$.
  \item \label{conn:factors} Suppose $B=B_1B_2$ with $B_1$,~$B_2 \in \cS$.
    Then there exists a completely pseudo-regular element $C \in [A]_\cS$ such that $B_1B_2= B_1E(C)B_2$.
  \end{enumerate}
\end{lemma}

\begin{proof}
  Let $P$, $P'$, $Q$,~$Q' \in \cS \cup \{I\}$ such that $A \parallel Q'BP'$ and $B \parallel QAP$.
  Let $r=\rank(A)$.

  \ref{conn:element}
  Since $\rank(B)=r$ as well, we have $\im(QA)=\im(B)$ and $\im(Q'B)=\im(Q'QA)=\im(A)$.
  Then $\rank(Q'QA)=r$ implies $\ker(Q'QA)=\ker(A)$, so that $Q'QA \parallel A$.
  In particular, $Q'QA$ is completely pseudo-regular.
  Now let $D\coloneqq QAQ'$.
  Since $\rank(Q'QAQ'QA)=r$, we must have $\rank(D)=r$.
  Then $\im(D)=\im(QA)=\im(B)$.
  Since $\rank(AQ'QA)=r$ we must have $\im(QA) \cap \ker(AQ')=0$, and thus $D$ is completely pseudo-regular.\
  Hence $E(D)B=B$.
  Finally, $D \parallel QAQ'$ by definition and $A \parallel Q'QAQ'QA=Q'DQA$, so that $A \sim_\cS D$.

  The symmetric claim follows analogously.

  \ref{conn:factors}
  By \ref{conn:element} there exist completely pseudo-regular elements $D$,~$D' \in [A]_\cS$ such that $B=E(D)B$ and $B=BE(D')$.
  Let $C \coloneqq (B_2D'P')A(Q'DB_1)$.
  From $\im(DB_1) \supseteq \im(DB)=\im(B)$ and $\rank(DB_1) \le \rank(B)$ we get $\im(DB_1)=\im(B)$.
  Analogously $\ker(B_2D')=\ker(BD')=\ker(B)$.
  Also $\im(Q'DB_1)=\im(Q'B)=\im(A)$ and $\ker(B_2D'P')=\ker(BP')=\ker(A)$.
  Thus $\rank(C)=r$.
  Computing $C^2=(B_2D'P')A(Q'DBD'P')A(Q'DB_1)$, we see that $C$ is completely pseudo-regular.
  From $A \parallel (Q'DB_1)C(B_2D'P')A$ we get $C \sim_\cS A$.

  From $\ker(DB_1)=\ker(C)$ we have $DB_1=DB_1E(C)$, and from $\im(B_2D')=\im(C)$ we have $E(C)B_2D'=B_2D'$.
  Thus $DB_1E(C)B_2=DB$ and $B_1E(C)B_2D'=BD'$.
  We deduce $B_1E(C)B_2|_{\im(D')} =B|_{\im(D')}$.
  Next $\ker (E(C)B_2) \subseteq \ker(DB) = \ker(B)$ implies $\ker(B_1E(C)B_2)=\ker(B)=\ker(D')$.
  So $K^d \!\!=\!\! \im(D')\!\oplus\! \ker(D')$, so $B_1E(C)B_2=B$.
\end{proof}

\begin{proposition} \label{p:closure-he}
  Let $E \in M_d(K)$ be idempotent of rank $r$ and let $H \subseteq E\overline{\cS}E$ be a closed subset.
  Then $\{\, A \in H : \rank(A) = r \,\}$ is contained in a subgroup of $\overline{\cS}$ \textup{(}with neutral element $E$\textup{)}, and it is possible to compute $\overline{\langle \{\, A \in H : \rank(A) = r \,\} \rangle}$.
\end{proposition}

\begin{proof}
  Let $V=\im E$.
  By a suitable change of basis, the endomorphisms of $V$ correspond to matrices with arbitrary entries in the upper left $r \times r$-block and zeroes everywhere else.
  The matrices $A \in H$ with $\rank A=r$ correspond to those matrices where the upper left $r\times r$-block is invertible, and all entries outside this block are zero.
  We may therefore compute $\overline{\langle \{\, A \in H : \rank(A) = r \,\} \rangle}$ by reducing to the invertible case (see \cref{sec:invertible}).
\end{proof}

In the following lemma keep in mind that if $A \in \cS$ is completely pseudo-regular, then the associated idempotent $E=E(A)$ may not be contained in $\cS$ but is always contained in $\overline{\cS}$ by \cref{l:closed}.

The, somewhat technical, statement \ref{subgroup-closure:H} ``connects'' the idempotent $F=F(B)$ of any completely pseudo-regular to $E$ in way that is needed for proving \ref{subgroup-closure:J}.
Statement \ref{subgroup-closure:H} will not be needed later on.

\begin{lemma} \label{l:subgroup-closure}
  Let $r \ge 0$ and let $A$ be a completely pseudo-regular element of $\cS$ of rank $r$.
  Suppose $Y \subseteq M_d(K)$ is a closed set with $X \cup \{\, B \in \cS : \rank(B) > r \,\} \subseteq Y$.
  Let $E \coloneqq E(A)$ and $H \coloneqq  \{\, B \in E Y^{\trianglelefteq 2 \binom{d}{r}+5} E : \rank(B)=r \,\} $.
  \begin{enumerate}
  \item \label{subgroup-closure:H} 
    If $F=E(B)$ for some completely pseudo-regular $B \in [A]_\cS$, then there exist $D \in Y^{\trianglelefteq \binom{d}{r}+2}E$ and $D^+ \in E\overline{\langle H \rangle} Y^{\trianglelefteq \binom{d}{r}+2}$ such that $D^+D=E$ and $DD^+=F$.
  \item \label{subgroup-closure:J} The set $\overline{\langle H \rangle}$ is computable and $Y^{\trianglelefteq \binom{d}{r}+2} \overline{\langle H\rangle}Y^{\trianglelefteq \binom{d}{r}+2}$ contains $[A]_\cS$.
  \end{enumerate}
\end{lemma}

\begin{proof}
  \ref{subgroup-closure:H}
  Recall $A=EA=AE$ and $\rank A=\rank B = \rank E = r$.
  Let $P$,~$Q \in \cS \cup \{I\}$ be such that $B \parallel QAP$.
  Then $\im(B)=\im(QAP)=\im(QA)=\im(QEA)$, with the middle equality holding because of $\rank(QA) \le r$.
  Also because of the ranks, therefore $\im(B)=\im(QE)$ and $\rank(QE)=r$.
  Analogously one finds $\ker(B) = \ker(EP)$ and $\rank(EP)=r$.
  Now $(EP)F=EP$ and $F(QE)=QE$.
  Write $P=P_1\cdots P_m$ and $Q=Q_1\cdots Q_n$ with $m$, $n \ge 0$ and $P_i$,~$Q_i \in X \cup \{\, B \in \cS : \rank(B) > r\,\}$.
  Choosing $m$, $n$ minimal, we get $\rank(P_iP_{i+1})=r$ for $i \in [1,m-1]$ and $\rank(Q_iQ_{i+1})=r$ for $i \in [1,n-1]$.
  Consider $EP_1\cdots P_m$ and $Q_1\cdots Q_nE$.
  Applying \ref{subproduct:parallel} of \cref{l:subproduct}, we find subproducts $D=Q_{i_1}\cdots Q_{i_k}E$ and $C=EP_{j_1}\cdots P_{j_l}$ with $k$,~$l \le \binom{d}{r}+2$ and such that $\im(D)=\im(F)$ and $\ker(C)=\ker(F)$.

  Now set $R\coloneqq CD$.
  Then $R \in H$.
  Therefore $\overline{\langle H \rangle}$ contains the pseudo-inverse $R^+=ER^+=R^+E$ satisfying $RR^+=R^+R=E$ by \cref{l:closed}.
  Define $D^+\coloneqq R^+C=R^+EC$.
  Then $D^+D = R^+CD=R^+R=E$.
  Furthermore $DD^+$ is idempotent with $\im(DD^+)=\im F$ and $\ker(DD^+)=\ker F$.
  Thus $DD^+=F$.

  \ref{subgroup-closure:J}
  One first computes $\overline{H}$ and then, using \cref{p:closure-he}, one can compute $\overline{\langle H \rangle}=\overline{\langle \{ C \in \overline H : \rank(C) = r \}\rangle}$ as a subset of $EM_d(K)E$.
  Note $E \overline{\langle  H \rangle} Y^{\trianglelefteq \binom{d}{r} + 2}X Y^{\trianglelefteq \binom{d}{r}+2} E \subseteq E \overline{\langle H \rangle} Y^{\trianglelefteq \binom{d}{r}+5} E = E \overline{\langle H \rangle} E  Y^{\trianglelefteq \binom{d}{r}+5} E$.
  Every element of this set having rank $r$ is also contained in $E \overline{\langle H \rangle} H E \subseteq \overline{\langle H \rangle}$.
  
  Let $B=B_1\cdots B_n \in [A]_\cS$ with $B_1$, $\ldots\,$,~$B_n \in X$.
  By \cref{l:conn}, there exist completely pseudo-regular elements $C_0$, $\ldots\,$,~$C_n \in [A]_\cS$ such that $B=E_0B_1E_1B_2\cdots E_{n-1}B_nE_n$ with idempotents $E_i=E(C_i)$.
  For each $E_i$, let $A_i \in Y^{\trianglelefteq \binom{d}{r}+2}E$ and $A_i^+ \in E\overline{\langle H\rangle}Y^{\trianglelefteq \binom{d}{r}+2}$ be such that $A_i^+A_i=E$ and $A_iA_i^+=E_i$ (these exist by \ref{subgroup-closure:H}).
  Then
  \[
    B = A_0(A_0^+ B_1 A_{1})(A_{1}^+ B_{2}A_{2}) \cdots (A_{n-1}^+B_nA_n)A_n^+.
  \]
  Each $A_i^+B_iA_{i-1}$ is contained in $\overline{\langle H \rangle}$ and $A_n^+ \in  E\overline{\langle H \rangle} Y^{\trianglelefteq \binom{d}{r}+2} \subseteq \overline{\langle H\rangle} Y^{\trianglelefteq \binom{d}{r}+2}$, so that we obtain $B \in Y^{\trianglelefteq \binom{d}{r}+2} \overline{\langle H \rangle} Y^{\trianglelefteq \binom{d}{r}+2}$.
\end{proof}

\subsection{Termination and correctness of \cref{alg:semigroupclosure}}

The following lemma forms the basis of the recursive strategy in \cref{alg:semigroupclosure}.
It reduces the problem of computing $\overline{\cS}$ to the computation of a suitable set of representatives of the completely pseudo-regular elements.

\begin{lemma} \label{l:strategy}
  Let $Y \subseteq M_d(K)$ be closed, $r \ge 0$, and suppose $Y$ contains $X \cup \{\, A \in \cS : \rank(A) > r\,\}$.
  \begin{enumerate}
  \item \label{strategy:class} If $B$ is completely pseudo-regular of rank $r$, then $[B]_\cS \subseteq \cT(Y,B)$.
  \item \label{strategy:recursion} If $T \subseteq M_d(K)$ is closed such that $Y \cup T$ contains every completely pseudo-regular $B \in \cS$ with $\rank(B) \ge r$, then
    \vspace{-0.75em}
    \[
      \{\, B \in \cS : \rank(B) \ge r \,\} \subseteq \overline{\big(Y \cup T\big)^{\le 2\binom{d}{r}+3}}.
    \]
  \end{enumerate}
\end{lemma}

The claim \ref{strategy:class} follows immediately from \ref{subgroup-closure:J} of \cref{l:subgroup-closure}. If some element of rank $> r$ is missing from $Y$, perhaps $[B]_\cS \not\subseteq \cT(Y,B)$, but $\cT(Y,B)$ is still computable.
We prove \ref{strategy:recursion} of \cref{l:strategy} after \cref{l:recursion-newcpr}.

Several things remain to check; in particular that \textsc{TryClose} will indeed succeed to compute the closure under certain assumptions on the sets $R_s$, that \textsc{FindCPR} will discover new completely pseudo-regular elements, and finally, that loops that increase the size of $R_s$ eventually terminate.

We need two final preparatory lemmas.
The first one allows us to find completely pseudo-regular elements.
This will be the key ingredient to make \textsc{FindCPR} work.

\begin{lemma} \label{l:recursion-newcpr}
  Let $r \ge 0$ and let $Y$, $T \subseteq M_d(K)$ be closed such that $X \cup \{\, B \in \cS : \rank(B) > r \,\} \subseteq Y$, and set
  \[
    Y' \coloneqq \overline{\big(Y \cup T\big)^{\le 2\binom{d}{r}+3}}.
  \]

  If there exists $A=A_1\cdots A_n \in \cS \setminus Y'$ with $A_1$, $\ldots\,$,~$A_n \in X$ and $\rank(A) \ge r$, then there exist $k < l$ such that the subproduct $A'=A_k\cdots A_l$ is completely pseudo-regular of rank $r$ and not contained in $Y \cup T$.
\end{lemma}

\begin{proof}
  Successively grouping together subproducts contained in $Y \cup T$, we find a representation $A_1\cdots A_n = C_1\cdots C_t$ with $C_i \in \langle A_1,\ldots,A_n \rangle \cap (Y \cup T)$ and $t$ minimal.
  By minimality of $t$, necessarily $C_k\cdots C_l \not \in Y \cup T$ for $k < l$.
  In particular, $\rank(C_k\cdots C_l)=r$.
  Since $A \not \in Y'$, necessarily $t \ge 2\binom{d}{r} + 4$.
 
  Now \ref{subproduct:cpr} of \cref{l:subproduct} implies that there exist $k < l$ such that $A'\coloneqq C_k\cdots C_l$ is completely pseudo-regular.
\end{proof}

\begin{proof}[Proof of \cref{l:strategy}, \ref{strategy:recursion}]
  Suppose the claim is false. Then there exists some $A \in \cS\setminus Y'$ with $\rank(A) \ge r$.
  Then \cref{l:recursion-newcpr} implies that there exists a completely pseudo-regular $B \in \cS \setminus (Y \cup T)$ with $\rank(B) \ge r$, contradicting our assumption.
\end{proof}

A second lemma allows us to bound the sizes of the sets $R_s$, and will ultimately yield termination of the algorithm.
Let $R \subseteq M_d(K)$ be a set of completely pseudo-regular matrices.
We define a directed graph $G(R)$, whose vertex set is $R$ and having a directed edge $A \to B$ if $\ker(B) \cap \im(A) = 0$.
(Loops are permitted, but this shall not make a difference in our considerations.)
In the following, \ref{scc:bound} should be compared to \cite[Proposition 8]{hrushovski-ouaknine-pouly-worrell18}.

\begin{lemma} \label{l:scc}
  \begin{enumerate}
  \item \label{scc:equivalent}
      If $A$,~$B \in G(R)$ are contained in the same strongly connected component \textup{(}SCC\textup{)}, then $A \sim_{\cS} B$.

    \item \label{scc:bound}
        The graph $G(R)$ has at most $\binom{d}{r}$ SCCs of rank $r$.
  \end{enumerate}
\end{lemma}

\begin{proof}
  \ref{scc:equivalent} Observe: if there is an edge $C \to D$ in $G(R)$, then $\ker(DC)=\ker(C)$ and $\rank(D) \ge \rank(DC) = \rank(C)$.
  So if $C$, $D$ are two elements of the same SCC, then $\rank(C)=\rank(D)$;
  if $C \to D$ is an edge, then also $\im(DC)=\im(D)$.

  Now let there be paths $A \to C_1 \to \cdots \to C_k \to B$ and $B \to D_1 \to \cdots \to D_l \to A$.
  Set $Q \coloneqq BC_k\cdots C_1A$ and $P\coloneqq AD_l\cdots D_1B$.
  Then $\im(QAP)=\im(B)$ and $\ker(QAP)=\ker(B)$, so that $B \parallel QAP$.
  Symmetrically, $A \parallel PBQ$.

  \ref{scc:bound}
  Let $A_1$, $\ldots\,$,~$A_k$ be vertices in distinct SCCs of rank $r$.
  Define $A_i \ge A_j$ if there is a path from $A_i$ to $A_j$.
  This relation is reflexive, transitive, and, since $A_i$ and $A_j$ are in distinct SCCs, anti-symmetric.
  Thus it is an order relation and we may reindex the matrices in such a way that there is no path from $A_j$ to $A_i$ if $j > i$.
  In particular, $\ker(A_i) \cap \im(A_j) \ne 0$ for $j > i$ and $\ker A_i \cap \im A_i = 0$.
  By \cref{l:key-finiteness}, $k \le \binom{d}{r}$.
\end{proof}

\begin{theorem} \label{t:main-hull}
  For a closed set $X \subseteq M_d(K)$ and $\cS = \langle X \rangle$, it is possible to compute $\overline{\cS}$.
\end{theorem}

\begin{proof}
  We show that \cref{alg:semigroupclosure} terminates and outputs $\overline{\cS}$.

  First note, in \textsc{TryClose}, the inclusions $X \subseteq Y_s \subseteq \overline{\cS}$ and $T_s \subseteq \overline{\cS}$ hold for all $s$.
  In particular $X \subseteq Y_1 \subseteq \overline{\cS}$.
  If \cref{alg:semigroupclosure} terminates, then $Y_1^2 \subseteq Y_1$, and so $Y_1 \subseteq \overline{\cS}$ is a closed overmonoid of $X$ contained in $\overline{\cS}$, so $\cS \subseteq Y_1 \subseteq \overline{\cS}$ and thus $Y_1 = \overline{\cS}$.
  Thus only the termination of the algorithm remains to be shown.
  We start with two observations.
  \begin{enumerate}[label=\alph*)]
  \item \label{obs:1} In \textsc{TryClose}, if $Y_{s+1}$ contains $\{\, B \in \cS : \rank(B) \ge s+1 \,\}$ and $Y_{s+1} \cup T_s$ contains all completely pseudo-regular elements of rank $s$, then $Y_s$ contains $\{\, B \in \cS : \rank(B) \ge s \,\}$ by \ref{strategy:recursion} of \cref{l:strategy}.
  Since this condition trivially holds for $s=\grank(X)$ (as $\{\, B \in \cS : \rank(B) \ge r+1 \,\} = \emptyset$), it suffices to construct the sets  $T_s$ so that $Y_{s+1} \cup T_s$ covers the completely pseudo-regular elements of rank $\ge s$, to obtain $\cS \subseteq Y_1$ inductively.

\item \label{obs:2}
  Throughout the algorithm, $R_s$ is a finite set of completely pseudo-regular elements of rank $s$.
  Further, if $\{\, A \in \cS : \rank(A) \ge s+1 \,\} \subseteq Y_{s+1}$, then the elements of $R_s$ are pairwise $\sim_\cS$-inequivalent.
  (This follows because any element added to $R_s$ is chosen outside of $T_s$ and \ref{strategy:class} of \cref{l:strategy}.)
  Then $\card{R_s} \le \binom{d}{s}$ by \cref{l:scc}.

  Conversely, if we ever end up with $\card{R_s} > \binom{d}{s}$ in the algorithm, we must have missed a completely pseudo-regular element of rank $> s$, and we search for such an element (loop at \cref{line:resetloop}).
  \end{enumerate}

  To show that the algorithm terminates, we now show:
  \begin{enumerate}
    \item \label{term:cpr}
  in \cref{line:cpr-call}, the call to \textsc{FindCPR} always returns a completely pseudo-regular element $B$ of $\cS$ of some rank $s' > s$, with $B$ not contained in $Y_{s'+1} \cup T_{s'}$;
  \item \label{term:loops} the loops in \cref{line:monoidloop,line:resetloop} terminate.
  \end{enumerate}

  \ref{term:cpr} When we call \textsc{FindCPR} there always exists $s' >s$ and $A \in \cS \setminus Y_{s'}$ with $\rank(A)\ge s'$:
  for the first iteration ($s=0$), the failed check on \cref{line:monoidloop} implies $\cS \not \subseteq Y_1$.
  In any other iteration, we have $\card{R_s} > \binom{d}{s}$, so $\{\, A \in \cS : \rank(A) \ge s+1 \,\} \not \subseteq Y_{s+1}$.

  Thus, in \textsc{FindCPR}, there exists $n \ge 0$ and $s' > s$ such that $X^n \setminus C_{s'} \ne \emptyset$, and the loop will eventually discover such a pair $(n,s')$.
  Then $n \ge 2\binom{d}{s'}+4$, as $X^{\le 2\binom{d}{s'}+3} \subseteq Y_{s'}$ (\cref{line:ys-init,line:def-ys}).
  We can pick such an element $A=A_1\cdots A_n \in X \setminus C_{s'}$ (on \cref{line:pickelt}) using \cref{l:prod-containment}.
  \Cref{l:recursion-newcpr} gives the existence of a completely pseudo-regular subproduct (chosen on \cref{line:pickcpr}).

  \ref{term:loops}
  Consider first the loop on \cref{line:resetloop}.
  In each iteration $s$ increases by at least $1$ (the rank of $B$ is larger then the value of $s$ passed to \textsc{FindCPR}).
  But at latest when $s=r$, we always have $\card{R_r} \le \binom{d}{r}$, by observation \ref{obs:2}, and the loop terminates.

  Consider now the outer loop, on \cref{line:monoidloop}.
  Outside of the loop on \cref{line:resetloop}, always $\card{R_s} \le \binom{d}{s}$ for all $s$ (inside the loop still $\card{R_s} \le \binom{d}{s}+1$).
  In each iteration we are increasing the size of some $R_s$ by one, while resetting all $R_{s'}$ with $s' < s$ to the empty set.
  Since $\card{R_r} \le \binom{d}{r}$ and $R_r$ is only ever growing, eventually $R_r$ must stabilize.
  Once this is the case, the algorithm does not modify $R_r$ any more and only touches the sets $R_{r-1}$, $\ldots\,$,~$R_1$.
  At this point $R_{r-1}$ can only ever grow.
  Thus, eventually, $R_{r-1}$ will also stabilize at $\card{R_{r-1}} \le \binom{d}{r-1}$.
  Inductively we conclude that eventually all the sets $R_{r-1}$, $\ldots\,$,~$R_1$ stabilize (there are no more new completely pseudo-regular elements to discover), and the algorithm stops.
\end{proof}

\begin{remark}[Efficiency] \label{rem:efficiency}
  \begin{enumerate}
  \item 
  While the algorithm largely works with linear algebra, and avoids the use of Gröbner bases (which can be computationally inefficient), the function \textsc{FindCPR} appears to be an obstacle to a reasonably efficient implementation.
  In particular, in the computation of elements in $X^n \setminus C_{s'}$, the exponent $n$ may become very large (there is no upper bound) and one needs to consider very long products of (generic) matrices.
  An obvious way of improving the algorithm, is therefore to find a better way of discovering the completely pseudo-regular elements.
\item In \textsc{FindCPR}, crucially, we choose the elements in $X^n \setminus C_{s'}$ instead of $\overline{X^n}\setminus C_{s'}$ (which would be nicer computationally), to avoid higher rank elements that may potentially appear in the closure (\cref{exm:generic-rank}).
\item We do not get runtime bounds.
  The problem is a lack of a bound for $n$ in \textsc{FindCPR}, and the lack of bounds on the number of steps in \cref{alg:groupclosure}.
\end{enumerate}
\end{remark}

\begin{remark}[Output size] \label{rem:upper-bound-output}
  For $X \subseteq M_d(K)$ closed, let $\comp{X}$ be the number of irreducible components of $X$.
  Let $\cS=\overline{\langle X \rangle}$.
  We sketch a double-exponential upper bound for $\comp{\cS}$ (and therefore also for the linear hull).
  We only consider $K=\bQ$.

  First consider the group case (i.e., $\GL_d(\bQ)$ is dense in $X$).
  In this case, we get a double-exponential bound in $d$ that does not depend on $X$:
  let $G \coloneqq \cS \cap \GL_d(\bQ)$ and let $G^0$ be the irreducible component containing $I$.
  We need to bound $\card{G/G^0}$.
  In \cref{t:alg1-correct}, we saw that $G/G^0$ is a subgroup of $\GL_{d'}(\bQ)$ for some $d'$.
  The embedding arises from applying \cite[Theorem II.6.8]{borel91}.
  Tracing through \cite{borel91}, in our linear setting, gives
  \[
    d' \le \left( \binom{d^2}{r}+d\right)^2 \le \big(2^{d^2}+d\big)^2 \le 4\cdot 4^{d^2}
  \]
  (for some $r$, using that the binomial coefficients sum to $2^{d^2}$).
  Finite subgroups of $\GL_{d'}(\bQ)$ have cardinality at most $2^{d'}d'!$ if $d'>10$ and for smaller $d'$ the maximal sizes are also known (\cite[Table 1]{berry-dubickas-elkies-poonen-smyth04})\footnote{This theorem of Feit depends on unpublished work. Friedland \cite{friedland97} gives a proof for large $d$. This yields \emph{some} double-exponential bound.}.
  So $\comp{\cS} \le 2^{4 \cdot 4^{d^2}} (4 \cdot 4^{d^2})!$ for all $d$.
  In general, one gets a bound that is double-exponential in $d$, by combining the group case with induction on the recursive strategy \cref{l:strategy} (the bound depends on $\comp{X}$).
\end{remark}

\bibliography{linear_hull}

\begin{thebibliography}{10}
\providecommand{\url}[1]{#1}
\csname url@samestyle\endcsname
\providecommand{\newblock}{\relax}
\providecommand{\bibinfo}[2]{#2}
\providecommand{\BIBentrySTDinterwordspacing}{\spaceskip=0pt\relax}
\providecommand{\BIBentryALTinterwordstretchfactor}{4}
\providecommand{\BIBentryALTinterwordspacing}{\spaceskip=\fontdimen2\font plus
\BIBentryALTinterwordstretchfactor\fontdimen3\font minus
  \fontdimen4\font\relax}
\providecommand{\BIBforeignlanguage}[2]{{%
\expandafter\ifx\csname l@#1\endcsname\relax
\typeout{** WARNING: IEEEtranS.bst: No hyphenation pattern has been}%
\typeout{** loaded for the language `#1'. Using the pattern for}%
\typeout{** the default language instead.}%
\else
\language=\csname l@#1\endcsname
\fi
#2}}
\providecommand{\BIBdecl}{\relax}
\BIBdecl

\bibitem{allauzen-mohri03}
C.~Allauzen and M.~Mohri, ``Efficient algorithms for testing the twins
  property,'' 2003, vol.~8, no.~2, pp. 117--144, weighted automata: theory and
  applications (Dresden, 2002).

\bibitem{alon99}
N.~Alon, ``Combinatorial {N}ullstellensatz,'' 1999, vol.~8, no. 1-2, pp. 7--29,
  recent trends in combinatorics (M\'{a}trah\'{a}za, 1995).

\bibitem{bell-smertnig21}
J.~Bell and D.~Smertnig, ``Noncommutative rational {P}\'{o}lya series,''
  \emph{Selecta Math. (N.S.)}, vol.~27, no.~3, pp. Paper No. 34, 34, 2021.

\bibitem{bell-smertnig23-arxiv}
------, ``Computing the linear hull: Deciding sequential? and unambiguous? for
  weighted automata over fields,'' 2023, arXiv version,
  \href{https://arxiv.org/abs/2209.02260}{arXiv:2209.02260}.

\bibitem{berry-dubickas-elkies-poonen-smyth04}
N.~Berry, A.~Dubickas, N.~D. Elkies, B.~Poonen, and C.~Smyth, ``The conjugate
  dimension of algebraic numbers,'' \emph{Q. J. Math.}, vol.~55, no.~3, pp.
  237--252, 2004.

\bibitem{berstel-mignotte76}
J.~Berstel and M.~Mignotte, ``Deux propri\'{e}t\'{e}s d\'{e}cidables des suites
  r\'{e}currentes lin\'{e}aires,'' \emph{Bull. Soc. Math. France}, vol. 104,
  no.~2, pp. 175--184, 1976.

\bibitem{berstel-reutenauer11}
J.~Berstel and C.~Reutenauer, \emph{Noncommutative rational series with
  applications}, ser. Encyclopedia of Mathematics and its Applications.\hskip
  1em plus 0.5em minus 0.4em\relax Cambridge University Press, Cambridge, 2011,
  vol. 137.

\bibitem{borel91}
A.~Borel, \emph{Linear algebraic groups}, 2nd~ed., ser. Graduate Texts in
  Mathematics.\hskip 1em plus 0.5em minus 0.4em\relax Springer-Verlag, New
  York, 1991, vol. 126.

\bibitem{bourbaki:ca72}
N.~Bourbaki, \emph{Elements of mathematics. {C}ommutative algebra}.\hskip 1em
  plus 0.5em minus 0.4em\relax Hermann, Paris; Addison-Wesley Publishing Co.,
  Reading, Mass., 1972, translated from the French.

\bibitem{buechse-vogler-may10}
M.~B\"{u}chse, H.~Vogler, and J.~May, ``Determinization of weighted tree
  automata using factorizations,'' \emph{J. Autom. Lang. Comb.}, vol.~15, no.
  3-4, pp. 229--254, 2010.

\bibitem{choffrut77}
C.~Choffrut, ``Une caract\'{e}risation des fonctions s\'{e}quentielles et des
  fonctions sous-s\'{e}quentielles en tant que relations rationnelles,''
  \emph{Theoret. Comput. Sci.}, vol.~5, no.~3, pp. 325--337, 1977.

\bibitem{colcombet-petrisan17}
T.~Colcombet and D.~Petri\c{s}an, ``Automata in the category of glued vector
  spaces,'' in \emph{42nd {I}nternational {S}ymposium on {M}athematical
  {F}oundations of {C}omputer {S}cience}, ser. LIPIcs. Leibniz Int. Proc.
  Inform.\hskip 1em plus 0.5em minus 0.4em\relax Schloss Dagstuhl.
  Leibniz-Zent. Inform., Wadern, 2017, vol.~83, pp. Art. No. 52, 14.

\bibitem{derksen-jeandel-koiran05}
H.~Derksen, E.~Jeandel, and P.~Koiran, ``Quantum automata and algebraic
  groups,'' \emph{J. Symbolic Comput.}, vol.~39, no. 3-4, pp. 357--371, 2005.

\bibitem{doerbrand-feller-stier21}
F.~D\"{o}rband, T.~Feller, and K.~Stier, ``Sequentiality of group-weighted tree
  automata,'' in \emph{Language and automata theory and applications}, ser.
  Lecture Notes in Comput. Sci.\hskip 1em plus 0.5em minus 0.4em\relax
  Springer, Cham, [2021] \copyright 2021, vol. 12638, pp. 267--278.

\bibitem{droste-kuich-vogler09}
M.~Droste, W.~Kuich, and H.~Vogler, Eds., \emph{Handbook of weighted automata},
  ser. Monographs in Theoretical Computer Science. An EATCS Series.\hskip 1em
  plus 0.5em minus 0.4em\relax Springer-Verlag, Berlin, 2009.

\bibitem{feit96}
W.~Feit, ``Orders of finite linear groups,'' in \emph{Proceedings of the
  {F}irst {J}amaican {C}onference on {G}roup {T}heory and its {A}pplications
  ({K}ingston, 1996)}.\hskip 1em plus 0.5em minus 0.4em\relax Univ. West
  Indies, Kingston, [1996], pp. 9--11.

\bibitem{filiot-gentilini-raskin15}
E.~Filiot, R.~Gentilini, and J.-F. Raskin, ``Quantitative languages defined by
  functional automata,'' \emph{Log. Methods Comput. Sci.}, vol.~11, no.~3, pp.
  3:14, 32, 2015.

\bibitem{friedland97}
S.~Friedland, ``The maximal orders of finite subgroups in {${\rm GL}_n({\bf
  Q})$},'' \emph{Proc. Amer. Math. Soc.}, vol. 125, no.~12, pp. 3519--3526,
  1997.

\bibitem{fulop-koszo-vogler21}
Z.~F\"{u}l\"{o}p, D.~K\'{o}sz\'{o}, and H.~Vogler, ``Crisp-determinization of
  weighted tree automata over strong bimonoids,'' \emph{Discrete Math. Theor.
  Comput. Sci.}, vol.~23, no.~1, pp. Paper No. 18, 44, 2021.

\bibitem{greuel-pfister08}
G.-M. Greuel and G.~Pfister, \emph{A {\bf {S}ingular} introduction to
  commutative algebra}, extended~ed.\hskip 1em plus 0.5em minus 0.4em\relax
  Springer, Berlin, 2008, with contributions by Olaf Bachmann, Christoph Lossen
  and Hans Sch\"{o}nemann.

\bibitem{herstein94}
I.~N. Herstein, \emph{Noncommutative rings}, ser. Carus Mathematical
  Monographs.\hskip 1em plus 0.5em minus 0.4em\relax Mathematical Association
  of America, Washington, DC, 1994, vol.~15, reprint of the 1968 original, With
  an afterword by Lance W. Small.

\bibitem{hrushovski-ouaknine-pouly-worrell18}
E.~Hrushovski, J.~Ouaknine, A.~Pouly, and J.~Worrell, ``Polynomial invariants
  for affine programs,'' in \emph{L{ICS} '18---33rd {A}nnual {ACM}/{IEEE}
  {S}ymposium on {L}ogic in {C}omputer {S}cience}.\hskip 1em plus 0.5em minus
  0.4em\relax ACM, New York, 2018, p.~10.

\bibitem{kirsten12}
D.~Kirsten, ``Decidability, undecidability, and {PSPACE}-completeness of the
  twins property in the tropical semiring,'' \emph{Theoret. Comput. Sci.}, vol.
  420, pp. 56--63, 2012.

\bibitem{kirsten-lombardy09}
D.~Kirsten and S.~Lombardy, ``Deciding unambiguity and sequentiality of
  polynomially ambiguous min-plus automata,'' in \emph{S{TACS} 2009: 26th
  {I}nternational {S}ymposium on {T}heoretical {A}spects of {C}omputer
  {S}cience}, ser. LIPIcs. Leibniz Int. Proc. Inform.\hskip 1em plus 0.5em
  minus 0.4em\relax Schloss Dagstuhl. Leibniz-Zent. Inform., Wadern, 2009,
  vol.~3, pp. 589--600.

\bibitem{kirsten-maeurer05}
D.~Kirsten and I.~M\"{a}urer, ``On the determinization of weighted automata,''
  \emph{J. Autom. Lang. Comb.}, vol.~10, no. 2-3, pp. 287--312, 2005.

\bibitem{kostolanyi22}
P.~Kostol\'{a}nyi, ``Determinisability of unary weighted automata over the
  rational numbers,'' \emph{Theoret. Comput. Sci.}, vol. 898, pp. 110--131,
  2022.

\bibitem{kuzmanovich-pavlichenkov02}
J.~Kuzmanovich and A.~Pavlichenkov, ``Finite groups of matrices whose entries
  are integers,'' \emph{Amer. Math. Monthly}, vol. 109, no.~2, pp. 173--186,
  2002.

\bibitem{lefaucheux-ouaknine-purser-worrell21}
E.~Lefaucheux, J.~Ouaknine, D.~Purser, and J.~Worrell, ``Porous invariants,''
  in \emph{Computer aided verification. {P}art {II}}, ser. Lecture Notes in
  Comput. Sci.\hskip 1em plus 0.5em minus 0.4em\relax Springer, Cham, [2021]
  \copyright 2021, vol. 12760, pp. 172--194.

\bibitem{lombardy-sakarovitch06}
S.~Lombardy and J.~Sakarovitch, ``Sequential?'' \emph{Theoret. Comput. Sci.},
  vol. 356, no. 1-2, pp. 224--244, 2006.

\bibitem{mohri97}
M.~Mohri, ``Finite-state transducers in language and speech processing,''
  \emph{Comput. Linguist.}, vol.~23, no.~2, pp. 269--311, 1997.

\bibitem{mohri09}
------, ``Chapter 6: {W}eighted automata algorithms,'' in \emph{Handbook of
  weighted automata}, ser. Monogr. Theoret. Comput. Sci. EATCS Ser.\hskip 1em
  plus 0.5em minus 0.4em\relax Springer, Berlin, 2009, pp. 213--254.

\bibitem{mohri-riley17}
M.~Mohri and M.~D. Riley, ``A disambiguation algorithm for weighted automata,''
  \emph{Theoret. Comput. Sci.}, vol. 679, pp. 53--68, 2017.

\bibitem{nosan-pouly-schmitz-shirmohammadi-worrell22}
K.~Nosan, A.~Pouly, S.~Schmitz, M.~Shirmohammadi, and J.~Worrell, ``On the
  {C}omputation of the {Z}ariski {C}losure of {F}initely {G}enerated {G}roups
  of {M}atrices,'' 2021, preprint.

\bibitem{okninski85}
J.~Okni\'{n}ski, ``Strongly {$\pi$}-regular matrix semigroups,'' \emph{Proc.
  Amer. Math. Soc.}, vol.~93, no.~2, pp. 215--217, 1985.

\bibitem{okninski91}
------, ``Linear representations of semigroups,'' in \emph{Monoids and
  semigroups with applications ({B}erkeley, {CA}, 1989)}.\hskip 1em plus 0.5em
  minus 0.4em\relax World Sci. Publ., River Edge, NJ, 1991, pp. 257--277.

\bibitem{okninski98}
------, \emph{Semigroups of matrices}, ser. Series in Algebra.\hskip 1em plus
  0.5em minus 0.4em\relax World Scientific Publishing Co., Inc., River Edge,
  NJ, 1998, vol.~6.

\bibitem{paul21}
E.~Paul, ``Finite sequentiality of unambiguous max-plus tree automata,''
  \emph{Theory Comput. Syst.}, vol.~65, no.~4, pp. 736--776, 2021.

\bibitem{putcha88}
M.~S. Putcha, \emph{Linear algebraic monoids}, ser. London Mathematical Society
  Lecture Note Series.\hskip 1em plus 0.5em minus 0.4em\relax Cambridge
  University Press, Cambridge, 1988, vol. 133.

\bibitem{renner05}
L.~E. Renner, \emph{Linear algebraic monoids}, ser. Encyclopaedia of
  Mathematical Sciences.\hskip 1em plus 0.5em minus 0.4em\relax
  Springer-Verlag, Berlin, 2005, vol. 134, invariant Theory and Algebraic
  Transformation Groups, V.

\bibitem{sakarovitch09}
J.~Sakarovitch, \emph{Elements of automata theory}.\hskip 1em plus 0.5em minus
  0.4em\relax Cambridge University Press, Cambridge, 2009, translated by Reuben
  Thomas.

\bibitem{stacks-project}
T.~{Stacks project authors}, ``The stacks project,''
  \url{https://stacks.math.columbia.edu}, 2019.

\end{thebibliography}

%



\clearpage
\appendix

\subsection{Other fields}

In the main text we restricted the field $K$ to be a number field (that is, a finite field extension of $\bQ$).
This restriction was made for simplicity of exposition.
In truth our approach does not impose restrictions on the nature of the field, except for the obvious necessity of the field being \emph{computable}, by which we mean (informally) that elements of the field can be represented exactly with finite memory, and equality comparisons between elements as well as the operations $+$, $\cdot$, $-$, $/$ can be computed exactly and in finite time.
This allows us to carry out linear algebra (Gaussian elimination) and computations with polynomials over such a field.

The rational numbers and finite fields of prime order (fields of the form $\mathbb F_p=\bZ/p\bZ$ with $p$ a prime number) are computable.
Finite-dimensional field extensions of computable fields are again computable when given by, e.g., generators and relations, or by a basis together with structure coefficients explaining the multiplication of basis elements.
Fields such as $\mathbb R$ or $\bC$ are not computable in this sense, however the field of algebraic numbers $\QQbar$ is computable (and implemented, for instance, in the SageMath computer algebra system).
The field $\bQ(\pi)$ is computable, because $\pi$ is transcendental and therefore $\bQ(\pi) \cong \bQ(x)$ is a rational function field.
The field $\bQ(\pi, e)$ is not known to be computable, because it is an open question in transcendence theory whether $\pi$ and $e$ are algebraically independent over $\bQ$.

When considering a weighted automaton, we may always work over fields that are \emph{finitely generated} (but not necessarily finite-dimensional) over their prime field ($\bQ$ or $\mathbb F_p$, depending on the characteristic).
Namely, we can take the field generated by all the entries of the vectors and matrices appearing in a linear representation of the automaton.
Let $K$ be a finitely generated field.
Then $K$ is a finite field, a number field, or a finitely generated extension of a finite or a number field $K_0$.
In the latter case, $K$ is the field of fractions of an affine $K_0$-algebra $R$.
We shall assume that $R$ is given by specifying generators and relations for $R$ over $K_0$.
That makes $R$, and therefore $K$, computable.

We now outline, section by section, which changes need to be made to deal with finitely generated fields.

\subsubsection{\Cref{sec:main}}
If $K$ is a finite field, then every vector space can be covered by a finite number of one-dimensional spaces (lines through the origin).
In this case, the irreducible closed sets in the linear Zariski topology are the vector spaces of dimension $\le 1$.
It follows that the linear hull always has dimension $\le 1$, and it becomes trivial to compute it.
As a consequence, one recovers the well-known result that a weighted automaton over a finite field is always determinizable.

If $K$ is an infinite field, the results in \cref{sec:main} remain valid as stated.

\subsubsection{\Cref{sec:onematrix}}

The algebraic closure $\QQbar$ of $\bQ$ has to be replaced by the algebraic closure $\algc{K}$ of $K$ throughout.
While the conclusion of \cref{t:one-matrix} remains true, several of the lemmas leading up to it, as well as the proof of \cref{t:one-matrix} itself, have to be adapted for the general case.

Write $\mu(\alg{K})$ for the group of all roots of unity.
If $\chr K = p > 0$, then $p$ does not divide the order of any root of unity.
For an integer $0 \ne n \in \bZ$, let $\val_p(n) \in \mathbb N_0$ denote the $p$-adic valuation, i.e., the number of times that $p$ divides $n$.

\begin{lemma} \label{l:general-invariant-no-roots-of-unity}
  Let $K$ be a field. Let $A \in \GL_d(K)$.
  \begin{enumerate}
  \item \label{l-g-inv1:coprime}
  Assume that for any two eigenvalues $\lambda$,~$\lambda' \in \alg{K}$ of $A$ for which $\lambda/\lambda' \in \mu(\algc{K})$, it holds that $\lambda=\lambda'$.
  Then a vector space $V \subseteq K^{d}$ is $A$-invariant if and only if it is $A^n$-invariant for all $n \ge 1$ with $\chr K \nmid n$.
\item \label{l-g-inv1:p}
    If $\chr K=p > 0$, and $V \subseteq K^d$ is $A^{p^n}$-invariant for some $n \ge 0$, then $V$ is $A^{p^e}$-invariant for $e=\val_p((d-1)!)$.
  \end{enumerate}
\end{lemma}

\begin{proof}
  Without restriction, assume $K=\algc{K}$.
  
  \ref{l-g-inv1:coprime}
  The proof is the same as the one of \cref{l:invariant-no-roots-of-unity}.
  The extra assumption $\chr K \nmid n$ (which is automatically satisfied if $\chr K=0)$ is necessary and sufficient for a primitive $n$-th root of unity $\zeta \in \algc{K}$ to exist.
  
  \ref{l-g-inv1:p}
  Using the direct-sum decomposition of $V$ along generalized eigenspaces, we can, as in the proof of \cref{l:invariant-no-roots-of-unity}, restrict to the case where $A$ has a single eigenvalue $\lambda$ (if $(\lambda/\lambda')^{p^n}=1$, then $\lambda=\lambda')$.
  Then $A - \lambda = N$ for some matrix $N$ with $N^d = 0$.
  Now
  \[
    A^{p^k} = (\lambda+N)^{p^k} = \sum_{i=0}^{\min\{d-1,p^k\}} \binom{p^k}{i} \lambda^{p^k-i}N^i.
  \]
  So, if $k \ge e \coloneqq \val_p((d-1)!)$, then  $\binom{p^k}{i}=0$ for $i \in [1,d-1]$ and $A^{p^k} = \lambda^{p^k}$.
  Thus, if $V$ is $A^{p^k}$-invariant for some $k \ge 0$, then it is $A^{p^e}$-invariant.
\end{proof}

\begin{lemma}
  Let $K$ be a finitely generated field.
  There exists a computable $N_0=N_0(d,K)$ such that, for every finite field extension $L/K$ with $[L:K] \le d$ and every root of unity $\zeta \in L$, one has $\zeta^{N_0}=1$ and moreover $\chr K \nmid N_0$.
\end{lemma}

\begin{proof}
  This makes essential use of the fact that $K$ is a finitely generated field.
  Suppose first $K=R=K_0$.
  Then either $K$ is a finite field, in which case the claim is trivial, or a number field, in which case the claim follows from \cref{l:computable-n0}.

  Now consider the general case.
  By effective Noether normalization \cite[Chapter 3.4]{greuel-pfister08}, we can compute transcendental $x_1$,~$\ldots\,$,~$x_n$ over $K_0$, such that $R$ is a finite module over $K_0[x_1,\ldots,x_n]$.
  Then $x_1$,~$\ldots\,$,~$x_n$ is a transcendence basis for $K/K_0$.
  From the generating set of $R$ as a $K_0[x_1,\ldots,x_n]$-algebra, we can compute a bound $m$ for the degree $[K:K_0(x_1,\ldots,x_n)]$.
  If $L$ is an extension of degree $d$ of $K$, then every element of $L$ that is algebraic over $K_0$ has degree $\le md$ over $K_0$.
  Thus we can take $N_0(d,K) = N_0(md,K_0)$.
\end{proof}

\begin{lemma} \label{l:general-invariant-exponent}
  Let $K$ be a finitely generated field.
  Let $p = \chr K$.
  Let $N \coloneqq N(d,K) \coloneqq  p^e N_0(d^2,K)$ with $e = \val_p((d-1)!)$.
  Let $A \in \GL_d(K)$, and let $V \subseteq K^{d}$ be a vector subspace.
  If $V$ is $A^n$-invariant for some $n \ge 1$, then $V$ is $A^N$-invariant.
\end{lemma}

\begin{proof}
  Let $\lambda$,~$\lambda' \in \alg{K}$ be eigenvalues of $A$ and let $N_0 = N_0(d^2,K)$.
  Since $\lambda$,~$\lambda'$ are both roots of the characteristic polynomial, which has degree $d$, the extension $K(\lambda,\lambda')/K$ has degree at most $d^2$.
  If there exists a root of unity $\zeta$ such that $\lambda/\lambda' = \zeta$, then $\zeta \in K(\lambda,\lambda')$ and hence $\zeta^{N_0} = 1$.
  Thus $A^{N_0}$ satisfies the assumption of \ref{l-g-inv1:coprime} of \cref{l:general-invariant-no-roots-of-unity}.

  Now suppose that $V$ is $A^n$-invariant with $n \ge 1$ and let $n = p^km$ with $k \ge 0$ and $m$ coprime to $p$.
  Replacing $n$ by a multiple of itself if necessary, we may assume $k \ge e$ and $N_0 \mid m$.
  Applying \ref{l-g-inv1:p} of \cref{l:general-invariant-no-roots-of-unity} to the matrix $A^m$ raised to the power $p^k$, the space $V$ is $(A^m)^{p^e}$-invariant.
  Using $(A^m)^{p^e}=(A^{p^e})^m$ and $N_0 \mid m$, we can now apply \ref{l-g-inv1:coprime} of \cref{l:general-invariant-no-roots-of-unity} to deduce that $V$ is $A^{p^eN_0}$-invariant.
\end{proof}

Now the proof of \cref{t:one-matrix} goes through as in the number field case, with \cref{l:invariant-exponent} replaced by \cref{l:general-invariant-exponent}.

\subsubsection{\Cref{sec:invertible}}

The proof of \cref{l:prod-containment} uses that $K$ is infinite, on the one hand to be able to find arbitrarily large subsets, and on the other to ensure that a nonzero polynomial does not vanish everywhere.
However, the conclusion of \cref{l:prod-containment} remains trivially true for finite fields.

The conclusion of \cref{l:prodclosure} is true over any field, but the stated proof requires the field to be infinite, to ensure that $V$, $W$ are also irreducible in the Zariski topology.
If $K$ is a finite field, and $V$ and $W$ are closed and irreducible subsets in the linear Zariski topology, then $V$ and $W$ are the zero space or one-dimensional vector spaces.
In the latter case, they are not irreducible in the Zariski topology (being a finite union of their finitely many points).
However, clearly $VW$ is again the zero space (if one of $V$ and $W$ is zero) or a one-dimensional space (if $V$ and $W$ are one-dimensional), so the conclusion of \cref{l:prodclosure} holds trivially.

\subsubsection{\Cref{sec:semigroup}} No changes are necessary.

\subsection{Integral domains that are not fields}

Suppose that $R$ is not a field but only a (commutative) domain (such as $\bZ$) and consider the problem of deciding determinizability and ambiguity for $R$-automata.
Of course, one can carry out the procedure over the quotient field $K=\quo R$ of $R$.
However the existence of a deterministic $K$-automaton equivalent to the initial one, may not imply the existence of a deterministic $R$-automaton.
Similar considerations apply for unambiguous automata.
Luckily, if $R$ is completely integrally closed we obtain the following.

\begin{corollary} \label{cor:main-unambiguous}
  Let $R$ be a finitely generated completely integrally closed domain and $\mathcal A$ an $R$-automaton.
  Then it is decidable if $\mathcal A$ is equivalent to an unambiguous $R$-automaton.
  In this case a corresponding unambiguous $R$-automaton is computable.
\end{corollary}

\begin{proof}[Sketch of proof]
  By \cite[Theorem 1.2]{bell-smertnig21}, the $R$-automaton $\cA$ is equivalent to an unambiguous $R$-automaton, if and only if $\cA$ is equivalent to an unambiguous $K$-automaton over $K$.
  The latter property can be decided by \cref{t:main-automata}.

  Suppose $\mathcal A'$ is an unambiguous $K$-automaton that is equivalent to $\cA$ (over $K$) and let $S \in R\langle\!\langle X \rangle\!\rangle$ be the corresponding rational series.
  Using \cite[Proposition 6.1]{bell-smertnig21} we get a representation of $S$ as an unambiguous $K$-rational series, and by \cite[Proposition 9.1]{bell-smertnig21} we obtain a representation as an unambiguous rational series over $R$, which yields an $R$-automaton.
\end{proof}

Unfortunately, passing through an unambiguous rational series as in the previous corollary, and back to an unambiguous automaton, it does not seem to be clear how to preserve the deterministic property.
However, if $R$ is a principal ideal domain (PID) there is a way to pass to $R$.

\begin{corollary} \label{cor:main-deterministic}
  Let $R$ be a finitely generated PID and $\mathcal A$ a $R$-automaton.
  Then it is decidable if $\mathcal A$ is equivalent to a deterministic $R$-automaton.
  In this case a corresponding deterministic $R$-automaton is computable.
\end{corollary}

\begin{proof}[Sketch of Proof]
  We claim that this again reduces to the same question over $K$.
  Clearly, if $\mathcal A$ is equivalent to a deterministic $R$-automaton, it is equivalent to a deterministic $K$-automaton.
  Suppose conversely that $\mathcal A$ is equivalent to a deterministic $K$-automaton.
  Then the linear hull of every minimal $K$-automaton is at most one-dimensional \cite[Theorem 1.3]{bell-smertnig21}.
  
  Let $(u,\mu,v)$ be a minimal linear representation of $\cA$ over $K$.
  By \cite[Theorem 7.1.1]{berstel-reutenauer11} we may assume that in fact $u \in R^{1 \times d}$, $\mu(w) \in M_d(R)$, and $v \in R^{d}$ for all $w \in \Sigma^*$.
  Let $\Omega \coloneqq \{\, u \mu(w): w \in \Sigma^* \,\}$.
  Now there are $a_1$, $\ldots\,$,~$a_n \in K^{1 \times d}$ such that $\Omega \subseteq (Ka_1 \cup \cdots \cup K a_n) \cap R^{ 1 \times d}$.
  We may take the coordinates of each $a_i$ to be in $R$ and to be coprime.
  Then $\Omega \subseteq Ra_1 \cup \cdots \cup R a_n$.
  This yields an $R$-deterministic automaton equivalent to $\cA$ (on $n$ states) \cite[Proposition 5]{lombardy-sakarovitch06}.
\end{proof}

\Cref{cor:main-unambiguous,cor:main-deterministic} apply to the ring of integers $\bZ$, and so Problem 1 of \cite{lombardy-sakarovitch06} also has a positive answer in this case.
The restriction to finitely generated domains is again so that basic computations (and the linear algebra in \cref{cor:main-deterministic}) can indeed be carried out.

\subsection{Derivation of bounds}

We sketch how to derive the bounds on the output size (\cref{rem:upper-bound-output}) in the case $K=\bQ$.
As a first step, we know that a quotient of a linear algebraic subgroup of $\GL_d$ by a normal subgroup is again linear algebraic (so, it can be be embedded in some $\GL_{d'}$).
This is a standard result in the theory of algebraic groups \cite[Theorem II.6.8]{borel91}, but unfortunately, making $d'$ explicit requires tracing through the proofs.
We sketch how to do this, following the proof in Borel's book \cite{borel91}.

We are ultimately interested in the $K$-rational points $G(K)$ of a linear algebraic group $G$, but to obtain the desired result, it is necessary to work in the language of algebraic geometry.
In a sense, this also means to consider the points $G(\algc{K})$ over the algebraic closure $\algc{K}$ of $K$.
However, the varieties will be defined over $K$ and morphisms will be $K$-morphisms (see \cite[\S 11]{borel91} for the precise definitions), so the results then descend to the group of $K$-rational points.

Consider the linear algebraic group $\GL_d$.
It is defined over our base field $K$, having the coordinate ring
\[
  K[\GL_d]=K[x_{ij}, \det(x_{ij})^{-1}] \cong K[x_{ij}, t] / (t\det(x_{ij})-1).
\]
(here $i$,~$j$ range over $[1,d]$).

The group $\GL_d(\algc{K})$ acts on $\algc{K}[\GL_d]$ by left translation, that is, for $g \in \GL_d(\algc{K})$ and $f \in \algc{K}[\GL_d]$, the action is defined by $(g,f) \mapsto \lambda_g f$ with $\lambda_g f(x) = f(g^{-1} x)$ for all $x \in \GL_d(\algc{K})$ \cite[\S II.1.9]{borel91}.
Here $g^{-1} x$ is just the usual the matrix product.

Let $L \subseteq \algc{K}[\GL_d]$ be the $d^2$-dimensional $\algc{K}$-vector space spanned by $\{\, x_{ij} : i, j \in [1,d] \,\}$.
(This vector space is \emph{defined over $K$}, in the sense of \cite[\S 11.1]{borel91}.)
Then $L$ is invariant under the $\GL_d(\algc{K})$-action: taking $f=x_{ij}$ and $g$, $y \in \GL_d(\algc{K})$ with $g^{-1}=(g_{ij}')$, $y=(y_{ij})$, we have
\[
  \lambda_g x_{ij} (y) = x_{ij}(g^{-1}y) = \sum_{\nu=1}^d g_{i\nu}' y_{\nu j} = \sum_{\nu=1}^d g_{i \nu}' x_{\nu j}(y),
\]
so $\lambda_g x_{ij} \in \lspan\{\, x_{1j}, \ldots, x_{dj} \,\} \subseteq L$.
By definition, the space $L$ contains all homogeneous linear polynomials in $x_{ij}$.

Now consider the case where $G \le \GL_d$ is a subgroup which is defined as the vanishing set of a set of homogeneous linear polynomials in the $x_{ij}$ and with coefficients in $K$ (this is the situation we are dealing with in \cref{sec:invertible}).
Solving the linear system, we obtain a subset $I \subseteq [1,d]^2$ such that the $\{\, x_{ij} : (i,j) \in I \,\}$ form a set of free variables for the system.

The coordinate ring $K[G]$ is he quotient of $K[\GL_d]$ by these equations.
We may think of it as
\[
  K[G] = K[x_{ij}, f(x_{ij})^{-1}] \cong K[x_{ij}, t]/(t f(x_{ij}) - 1), 
\]
where now $(i,j) \in I$ and $f(x_{ij})$ is a polynomial in $x_{ij}$ with $(i,j) \in I$, obtained from the determinant by substituting the solution of the linear system.
As in the case $G=\GL_d$ before, the group $G(\algc{K})$ acts on $\algc{K}[G]$ by left translation, and we have a $G(\algc{K})$-invariant vector subspace $L_G$ spanned by $\{\, x_{ij} : (i,j) \in I \,\}$ and with $\dim(L_G) = \card{I} \le d^2$.

The following is a version of \cite[Theorem II.6.8]{borel91}, restricted to our setting, that gives an explicit bound on the dimension of a matrix group that $G/N$ can be embedded in.

\begin{proposition}
  Let $G \le \GL_d$ be a $K$-subgroup.
  Let $N \subseteq G$ be a closed normal $K$-subgroup, defined as a $K$-variety in $G$ by homogeneous linear polynomials in the matrix entries $\{\, x_{ij} : (i,j) \in I \,\}$, with coefficients in $K$.
  Then there exists $r \in [1,d]$ such that $G/N$ is an affine $K$-subgroup of $\GL_{d'}$ and
  \[
    d' \le \left( \binom{d^2}{r} + d \right)^2 \le (2^{d^2} + d)^2.
  \]
\end{proposition}

\begin{proof}[Sketch of Proof]
  The vanishing ideal $J \subseteq \algc{K}[G]$ of $H$ is generated by the homogeneous linear polynomials defining $H$.
  Therefore, the finite-dimensional $G$-invariant subspace $L_G$ of $\algc{K}[G]$ contains this generating set of $J$.
  Let $W = L_G \cap J$ and $r = \dim(W)$.
  Put $E = \left( \bigwedge^r V  \right) \oplus (\algc{K})^{d}$.
  Then
  \[
    \dim_{\algc{K}}(E) = \binom{\card{I}}{r} + d \le \binom{d^2}{r} + d.
  \]
  Following the proof of \cite[Theorem II.5.1]{borel91}, this gives an immersive representation $\alpha: G \to \GL(E)$ (defined over $K$) and a line $D=\bigwedge^r V \subseteq E$ satisfying the conclusions of \cite[Theorem II.5.1]{borel91} with respect to $H=N$.
  As in \cite[Theorem II.5.6]{borel91}, this can be improved to $N=\ker(\alpha)$ (and the analogous condition $\mathfrak n = \ker(d\alpha)$ on the associated derivation), by replacing $E$ by a subspace $E'$ of $\GL(E)$ (cf. the third paragraph of the proof of \cite[Theorem II.5.6]{borel91}).
  Then $\dim(E') \le \dim(E)^2$.

  Finally, the proof of \cite[Theorem II.6.8]{borel91} shows that $G/N$ is an affine $K$-subgroup of $\GL(E')$.
\end{proof}

Since the morphism in the previous proposition is a $K$-morphism, it gives rise to an embedding of $K$-rational points $(G/N)(K) \subseteq \GL_{d'}(K)$.
In the output of \cref{alg:groupclosure}, the subgroup $N$ is the irreducible component of $G$ containing the identity, and $G/N$ is finite.
The number of irreducible components of the output is $\card{(G/N)(K)}$.
We have now seen that $(G/N)(K)$ is a finite subgroup of $\GL_{d'}(K)$ with explicitly bounded $d'$, so it suffices to bound the size of finite subgroups of $\GL_{d'}(K)$.

To do so, we now restrict to $K=\bQ.$
By a theorem of Feit \cite{feit96}, finite subgroups of $\GL_{d'}(\bQ)$ have cardinality at most $2^{d'}d'!$ if $d'>10$.
For $d' \le 10$, Feit also classified the finite subgroups of maximal cardinality \cite[Table 1]{berry-dubickas-elkies-poonen-smyth04}.
Unfortunately, the theorem of Feit depends on unpublished work of Weisfeiler (see the introduction of \cite{friedland97} or \cite[\S 5, \S 6]{kuzmanovich-pavlichenkov02} for a discussion).
Let $X \subseteq M_d(K)$ be closed such that $X \cap \GL_d(K)$ is dense in $X$.
Set $\cS = \langle X \rangle$.
Under the assumption that the Feit result holds, one obtains
\[
   \comp{\overline{\cS}} \le 2^{4 \cdot 4^{d^2}} (4 \cdot 4^{d^2})!
\]
for all $d$ by bounding $(2^{d^2} + d)^2 \le 4 \cdot 4^{d^2}$.
(This also works for $d \le 10$ because the bound is sufficiently large compared to the cardinalities of finite subgroups listed in \cite[Table 1]{berry-dubickas-elkies-poonen-smyth04}.)
Of course this bound is not sharp, e.g., for $d=1$ it gives $\approx 1.3 \cdot 10^{18}$, whereas in this case actually $\comp{\overline{\cS}} = 1$.

Avoiding the use of unpublished work, independently of the theorem of Feit, Friedland \cite{friedland97} uses a different (published) result of Weisfeiler to show that a finite subgroup of $\GL_d(\bQ)$ has cardinality $\le 2^d d!$ for all sufficiently large $d$.
From this result one gets the existence of \emph{some} double-exponential bound for $\comp{\overline{\cS}}$, but not an explicit one.
In any case, Weisfeiler's results, and hence these bounds, depend on the classification of finite simple groups.

To extend a bound to different fields $K$, it would be necessary to understand the maximal cardinality of finite subgroup of $\GL_{d'}(K)$.

\smallskip
\textbf{Semigroup case.}
In the general (semigroup) case we get a bound on the output size by combining the bound for the group case with the recursive strategy of \cref{l:strategy}.
Here it is no longer possible to obtain a bound that is independent of the size of the input set (and that only depends on the dimension $d$).
To see this, consider a finite subset $M \subseteq K$ and let
\[
  X = \bigcup_{m \in M} \lspan\left\{ \begin{pmatrix} 1 & m \\ 0 & 0 \end{pmatrix} \right\} \subseteq M_2(K),
\]
which is a union of $\card{M}$ pairwise distinct one-dimensional vector spaces, so $\comp{X}=\card{M}$.
One checks easily that $X$ is a semigroup.

Let us start with some easy observations: if $X$, $Y \subseteq M_d(K)$ are closed sets, then $\comp{\overline{XY}} \le \comp{X}\, \comp{Y}$.
Thus
\[
  \comp{\overline{X^{\le n}}} \le \sum_{i=1}^n \comp{X}^i \le n \,\comp{X}^{n} \le \comp{X}^{n+1}, 
\]
and also $\comp{\overline{X^{\trianglelefteq n}}} \le 1 + n \,\comp{X}^n \le \comp{X}^{n+1}$.

Now we can bound the size of the sets $\cT(Y,A)$ (page~\pageref{def:tya}):
Assume that $C(d)$ is the maximal size of a finite subgroup of $M_d(K)$.
Then $\comp{\overline{\langle \cT_0(Y,A) \rangle}} \le C(r)$ (with $r=\rank(A)$), and
\[
  \comp{\cT(Y,A)} \le C(r)\, \comp{Y}^{2\binom{d}{r}+4}.
\]

Looking at \textsc{TryClose} and keeping in mind \cref{l:scc}, we get 
\[
  \comp{T_s} \le \binom{d}{s} C(s) \, \comp{Y_s}^{2\binom{d}{r} + 4} \le 2^d C(d) \, \comp{Y_s}^{2^d + 4},
\]
and
\[
  \begin{split}
    \comp{Y_s} &\le \big(\comp{Y_{s+1}} + \comp{T_s}\big)^{2\binom{d}{s}+3} \\
    &\le \big(\comp{Y_{s+1}} + 2^d C(d) \, \comp{Y_s}^{2^d + 4}\big)^{2^d+3}.
  \end{split}
\]
Suppose $C(d)$ satisfies a double-exponential bound, i.e., $C(d) \le 2^{2^{Q(d)}}$ for some polynomial $Q(s)$.
Then also $\comp{Y_s}$ satisfies a double exponential bound, i.e.,
\[
  \comp{Y_s} \le \comp{Y_{s+1}}^{2^{P_0(d)}},
\]
for a suitable polynomial $P_0(d)$ (which does not depend on $Y_{s+1}$).
Inductively, we get
\[
  \comp{\overline{\cS}} = \comp{Y_1} \le \comp{X}^{2^{(d-1)P_0(d)}}.
\]

So altogether we obtained the following.

\begin{proposition}
  If $X \subseteq M_d(\bQ)$ is a closed set and $\cS = \langle X \rangle$, then the number of components $\comp{\overline{\cS}}$ of $\overline{\cS}$ can be bounded by
  \[
    \comp{\overline{\cS}} \le \comp{X}^{2^{P(d)}},
  \]
  with $P(d)$ a suitable polynomial.
  A similar upper bound holds for the number of components of the linear hull of a $\bQ$-automaton.
\end{proposition}

The conclusion holds over any field $K$ where one has a bound on cardinality of a finite subgroup of $\GL_d(K)$ that is double-exponential in $d$.

\end{document}